\documentclass[pdflatex,sn-mathphys]{sn-jnl}
\usepackage{amssymb}
\usepackage{mathtools}
\usepackage{booktabs}
\usepackage{braket}
\usepackage{amsthm}
\usepackage{macros}
\usepackage{enumitem}
\usepackage{derivative}
\usepackage{mathtools}
\usetikzlibrary{calc}
\usepackage{natbib}

\jyear{2021}%

\theoremstyle{thmstyleone}%
\newtheorem{thm}{Theorem}[section]
\newtheorem{prop}[thm]{Proposition}
\newtheorem{lem}[thm]{Lemma}
\newtheorem{cor}[thm]{Corollary}
\theoremstyle{thmstyletwo}%
\newtheorem{examplenn}[thm]{Example}%
\newtheorem{rmk}[thm]{Remark}%
\theoremstyle{thmstylethree}%
\newtheorem{defn}[thm]{Definition}
\newtheorem{problem}[thm]{Problem}%
\numberwithin{equation}{section}

\begin{document}

\title[Estimates on the ImFT]{Estimates on the size of the domain of the Implicit Function Theorem: A Mapping Degree based Approach}


\author*[1]{\fnm{Ashutosh} \sur{Jindal}}\email{ashutosh.1@sc.iitb.ac.in}

\author[1]{\fnm{Debasish} \sur{Chatterjee}}\email{dchatter@iitb.ac.in}

\author[1]{\fnm{Ravi} \sur{Banavar}}\email{banavar@iitb.ac.in}

\affil*[1]{\orgdiv{Systems and Control Engineering}, \orgname{Indian Institute of Technology, Bombay}, \orgaddress{\street{IIT Area, Powai}, \city{Mumbai}, \postcode{400076}, \state{Maharashtra}, \country{India}}}


\abstract{In this article we present explicit estimates of size
of the domain on which the Implicit Function Theorem and  the Inverse Function Theorem are valid. For maps that are \correction{twice continuously differentiable}, these estimates depend upon the magnitude of the first-order derivatives evaluated at the point of interest, and a bound on the second-order derivatives over a region of interest. One of the key contributions of this article is that the estimates presented require minimal numerical computation. In particular, these estimates are arrived at without any intermediate optimization procedures. We then present three applications in \correction{optimazation and} systems and control theory where the computation of such bounds \correction{turns out to be important.} \correction{First,} in electrical networks, the power flow operations can be written as Quadratically Constrained Quadratic Programs (QCQPs), and we utilize our  bounds to compute \correction{the size of permissible} power variations to ensure stable operations of the power system network. \correction{Second, robustness margin of positive definite solutions to the Algebraic Riccati Equation} (frequently encountered in control problems) subject to perturbations in the system matrices \correction{are} computed \correction{with the aid of our bounds}. \correction{Finally, we employ these bounds to provide quantitative estimates of the \correction{size of the} domains for feedback linearization of discrete-time control systems.}
}




\maketitle

\section{Introduction}\label{sec1}
The \emph{Implicit Function Theorem} (ImFT) and \correction{its conjoined twin} the \emph{Inverse Function Theorem} (IFT) constitute a cornerstone of mathematical analysis and multivariate calculus \cite{krantz2002implicit}. \correction{These theorems serve as the basis of several existential results in mathematics and find applications in the following areas:  optimization   \cite{bertsekas1999nonlinear,dontchev1999lipschitzian,fletcher2001practical}, numerical analysis,
\cite{mordukhovich1994lipschitzian,nash1996linear}, control theory applications \cite{lhmnc}, ordinary differential equation \cite{eldering2013normally,duistermaat2012lie}, etc. In control theory applications specifically, the ImFT allows us to  show the existence of solutions of ordinary differential equations \cite{eldering2013normally,duistermaat2012lie}, to assert the existence of a unique trajectory for the dynamical system, and also  continuity properties of the solutions with respect to initial conditions and the control input.} \correction{In particular, it appears as the driving engine behind several assertions concerning the so called \emph{end point map} \cite[Chapter~2,~3]{bressan2007}.} ImFT also finds applications in areas such as feedback linearization of numerically discretized systems \cite{lhmnc}. 

\correction{The assertions made by these theorems are of existential and local character, i.e., they hold only in a \correction{sufficiently small} neighborhood around the point of interest. This feature extends to the methods developed based on these results. For instance, the well-known Frobenius's Theorem\footnote{Here we refer to the local version of the theorem, however a global version of the Frobenius theorem also exists.} utilizes the IFT to compute a coordinate system that rectifies a given involutive distribution \cite{isidori1985nonlinear}}, \correction{and the region of such rectification is only qualitatively available in general.} 
\correction{The ImFT does not provide us with information on the size of the domain on which these results are valid. In engineering applications, having a quantitative estimate of the size of the domain is useful and  often crucial. Such a quantitative analysis of the ImFT and the IFT with explicit bounds on the sizes of the respective domains on which the theorems are applicable would be of use in many existential results in the broad area of cybernetics, especially where estimates of robustness are essential. This need
serves as our chief motivation to arrive at estimates of the domain of the validity of the ImFT and the IFT.}

Despite the wide applicability of the ImFT and IFT, few attempts have been made on estimating such
quantitative bounds of the domains of validity of these theorems. One set of such estimates for the ImFT was provided by \correction{Holtzman~\cite{holtzman1970explicit}}, and the estimates there are explicit functions of the bounds on the first-order derivatives of the underlying map over a given domain of interest. \correction{Chang et al.~\cite{chang2003analytic}}  provide another set of estimates on the  neighborhoods involved in ImFT  based on the application of the Roche Theorem \cite{ash2014complex}. These estimates were first provided for the scalar case and then applied inductively to generalize it for vector-valued maps. The bounds provided by \correction{Chang et al.~\cite{chang2003analytic}} were based on the boundedness of the underlying complex map over a bounded domain. The accuracy of these estimates is contingent on the accuracy with which one computes the bound of the underlying map over the domain of interest and for vector-valued maps, the bounds were calculated component-wise using induction. \correction{The method needs computation of  $(m-1)\times (m-1)$ sub-determinants of a certain Jacobian matrix, where $m$ is the dimension of the underlying vector space, and since the bounds are an explicit function of these subdeterminants, they must be computed over all possible minors, due to which the method performs poorly with an increase in the dimension.} \correction{In the context of the ImFT for Lipschitz continuous maps}, a third set of bounds based on the \correction{bounds on the generalized derivatives of the underlying map computed over a domain of interest}, was provided by \correction{Papi~\cite{papi2005domain}} for the ImFT for Lipschitz continuous functions. Although the bounds provided by Papi in \correction{\cite{papi2005domain}}  cover a larger class of functions than those covered by \correction{Chang et al.~\cite{chang2003analytic}}, to compute these bounds one needs to ensure that the generalized derivatives of the underlying map evaluated at the point of interest are invertible, and the accuracy of the estimate is dependent on how tightly the bounds on the generalized derivatives are computed over a given domain.

For the IFT, \correction{Abraham et al.~\cite{abraham2012manifolds}} provide a set of bounds based on the magnitude of the first-order derivatives evaluated at the point of interest and the boundedness of the magnitude of second-order derivatives over a domain. The boundedness of the second-order derivative restricts the application of these bounds to functions that are continuously differentiable at least up to second order. These bounds require minimal numerical computations since one only needs to compute the bounds on the second-order derivatives, and thus are useful in situations where limited computation capacity is available \correction{and/or coarse structural information about the maps are at hand.} 

In this article, we utilize the topological degree to compute estimates of the neighborhoods given by the IFT and ImFT. 
Degree theoretic methods \correction{are by nature flexible, depend on a few basic properties of the underlying maps,} and require minimal assumptions on the premise; in particular, one only requires continuity of various maps over the underlying sets, \correction{ although for computational ease one can utilize higher order differential properties. The technique is inherently topological and permits appreciable flexibility in terms of appropriate homotopies (as we shall see in the sequel).} This makes the analysis applicable to a relatively large class of functions. \correction{Although the key results presented in this article are given for $\cont{2}$ functions, these bounds are further generalized for $\cont{1}$ (see Proposition \ref{Imft_holtzman} in Section \ref{key_res}) and $\cont{0}$ functions (see Proposition \ref{imft-c0} in Appendix \ref{appen1}).} 
\subsection{Contributions}
\begin{enumerate}[label=\textup{(\roman*)},leftmargin=*, widest=b, align=left]
\item We utilize the topological degree to compute lower bounds on the domain of  validity of the ImFT.  Note that our objective here is to arrive at bounds that require \textbf{minimal numerical computation}. For $\cont{2}$ maps, our bounds are dependent on the first-order derivatives evaluated at the point of interest and second-order derivatives on a bounded set containing \correction{the} point of interest. These estimates are reported in Theorem \ref{imft_est}. 
\item Using Theorem \ref{imft_est}, in Corollary \ref{ift_est}, we derive estimates on the size of the domain of validity of the IFT. \correction{For finite-dimensional spaces, these estimates \textbf{improve on} the estimates given by Abraham et al.~\cite[Proposition 2.5.6]{abraham2012manifolds}.}
\item \correction{We also demonstrate that several existing results on the estimates on the domain of applications of IFT and ImFT can also be derived by degree theoretic methods. In particular, we show that when restricted to finite-dimensional vector spaces, the bounds given Holtzman~\cite{holtzman1970explicit} (see Proposition \ref{Imft_holtzman} in Section \ref{key_res}) and  Abraham et al.~\cite{abraham2012manifolds} (see Proposition \ref{propamd} in Appendix \ref{appen}) can also be derived by elementary application of the mapping degree} 
\item As the last theoretical contribution of this article, in Proposition \ref{imft-c0}, we extend these bounds to the generalized ImFT for continuous maps. (However, uniqueness and regularity properties cannot be asserted in this framework.) 
    %
    \item \correction{We present two key applications of our bounds: As our first application, we investigate the robustness of the solutions of Quadratically Constrained Quadratic Problems. In this setting we present two examples: one is the robustness margin for the stable operation of a power system network, and the other is the robustness of the solutions of the Algebraic Riccati Equation, a nonlinear algebraic equation, frequently encountered in optimal control
    theory. In the second application we utilize our bounds to estimate the domain on which a given discrete time control system is feedback linearizable.}
\end{enumerate}
\correction{Finally, through this work we also intend to draw the attention of the readers to the gamut of tools offered by the mapping degree theory; we selected the ImFT and IFT because of their central importance in much of control theory.}
\subsection{Notations}   We use standard notations throughout the article. The set of real numbers is denoted by $\R$ and the set of integers is denoted by $\Z$. The set of positive integers is denoted by $\N$, and the set of all positive integers less than or equal to $n$ is denoted by $\n$.
\\ \\ 
For $U\subset \R[n]$ a nonempty set, $\cl{U}$ denotes the smallest closed set containing $U$ called the closure of $U$; $\inte{U}$ is the largest open set contained in $U$ called the interior of $U$; and $\bd{U} = (\cl{U})\setminus\inte{U}$ denotes the boundary of $U$. For a given $r>0$ and $x_0\in \R[n]$,$$\ball{x_0}{r}=\{x\in\R\mid\norm{x-x_0}<r \}$$ denotes the open ball of radius $r$ centered \correction{at} $x_0$, for some fixed well-defined norm $\norm{\cdot}$ on $\R[n]$. The identity matrix of order $n$ is denoted by $\eye{n}$; if the order is unambiguous from the expression, we may drop the subscript and simply write $\eye{}$. For given $n,m\in\N$, for $A\: \R[n]\lra\R[m]$ a linear map,  $\norm{A}$ is defined by
$$\norm{A} = \supr{\norm{Ax}}{x\in\correction{\cl}\ball{0}{1}\subset\R[n]}.$$
For a given $\r\in \N$ and nonempty set $U$ and $V$, $\contr{\r}{U}{V}$ denotes the class of \emph{$\r$-times continuously differentiable maps} with domain and co-domain as $U$ and $V$\correction{,} respectively. If $U$ and $V$ are unambiguous from the context, we may simply write $\cont{\r}$. The differential operator is denoted by $\D$. For a given $f\in \contr{\r}{U}{\R[m]}$, mapping $\R[n]\supset U\in x\mapsto f(x)\in\R[m]$, $x_0\in U$, $\D f(x_0)$ is the Jacobian matrix of $f$ evaluated at $x_0$. The partial derivatives are denoted by adding subscript to the differential operator i.e., $\pdv{}{x}\eqqcolon \pdf{}{x}$.
\section{Degree Theory: Preliminaries}
\label{deg_theory}
In order to derive the estimates for the domains of the ImFT and IFT, we utilize several results from topological degree theory. This section serves as a rapid refresher on the topological degree. For a detailed study, one may look into \cite{zeidler1993vol,outerelo2009mapping,dinca2021brouwer}. 
\begin{defn}(Topological Degree \cite[Chapter~IV,~Proposition and Definition 1.1]{outerelo2009mapping})
\label{top_deg}
\correction{
Let $U\subset\R[n]$ be a nonempty and bounded open set and $f\: \cl U\lra\R[n]$ be a smooth map. Suppose $0\in\R[n]\setminus f(\bd U)$ is a regular value of $f$. Then $\eqpoint[f] \coloneqq f^{-1}(0)$ is finite (possibly empty), and we define the degree of $f$ by 
\begin{equation}
\label{deg_f}
    \DEG{f}{U}{} = \sum_{x\in \eqpoint[f]}\sign(\det (\D f(x)))
\end{equation}
where $\R\ni z\mapsto \sign(z)\in\{-1,0,1\}$ is defined as 
\begin{equation*}
    \sign(z) = 
    \begin{cases}
        +1\quad &z>0,\\
        0\quad &z=0,\\
        -1\quad &z<0.
    \end{cases}
\end{equation*}
}
\end{defn}
\begin{defn}(Topological Degree (for continuous maps) \cite[Chapter~IV,~Proposition and Definition 2.1]{outerelo2009mapping}) \correction{Let $U\subset\R[n]$ be bounded and open. Let $f\:\cl U \lra\R[n]$ be a continuous map. Suppose $0\in\R[n]\setminus f(\bd U)$, then there exists a smooth mapping $g\: \cl U \lra \R[n]$ such that $0$ is a regular value of $g$ and for all $x\in\bd U$, $\norm{f(x)-g(x)}<\norm{g(x)}$. For all such $g$'s, $\DEG{g}{U}$ is same and we define 
\begin{equation}
    \DEG{f}{U} \coloneqq \DEG{g}{U}.
\end{equation}}
\end{defn}
\subsection{Properties of $\DEG{f}{U}$} \correction{Let U and $f$ be as in \ref{top_deg}, then $\DEG{f}{U}$ satisfies the following properties.}  
\begin{enumerate}[label=\textup{(\textbf{DEG}-\alph*)},leftmargin=*, widest=b, align=left]
    \item \textbf{Negative Existence Principle} \cite[Chapter 4~Corollary 2.5]{outerelo2009mapping}: If $f$ is nonvanishing on $\cl U$, i.e., $\eqpoint[f] = \emptyset$ then $\DEG{f}{U}=0$.
    \label{deg_NI}
    \end{enumerate}
    \begin{rmk}
    An immediate consequence of \ref{deg_NI} is: $\DEG{f}{U}\neq 0$ implies $\eqpoint[f]\neq \emptyset$ i.e., there exists at least one $x\in U$ such that $f(x)=0$. 
    \end{rmk}
    \begin{enumerate}[resume,label=\textup{(\textbf{DEG}-\alph*)},leftmargin=*, widest=b, align=left]
    \item \textbf{Additivity} \cite[Chapter 4~Corollary 2.5]{outerelo2009mapping}: Let $U_1,U_2\subset U$ be nonempty and bounded open sets. Suppose \correction{$U_1\cap U_2 =\emptyset$} and $f$ is non vanishing on $U\setminus\cl{(U_1\cup U_2)}$, then $$\DEG{f}{U} = \DEG{f}{{U_1}}+\DEG{f}{{U_2}}.$$
    \label{deg_AD}
\end{enumerate}
\begin{defn}(Nonvanishing Homotopy)
\correction{
Let $U\subset\R[n]$ be bounded and open. Let $[0,1]\times \cl U \ni (t,x)\mapsto H(t,x)\in\R[n]$ be continuous. Moreover, for each $t\in[0,1]$, $0\in\R[n]\setminus H(t,\bd U)$, then $H$ defines a nonvanishing homotopy on $\bd U$. Two given continuous maps $f_1\:\cl U\lra\R[n]$ and $f_2\:\cl U\lra\R[n]$ are said to be nonvanishingly homotopic on $\bd U$ if there exists a nonvanishing homotopy $H$ such that $H(0,\cdot) = f_1$ and $H(1,\cdot)= f_2$.}        
\end{defn}
\begin{enumerate}[resume,label=\textup{(\textbf{DEG}-\alph*)},leftmargin=*, widest=b, align=left]
    \item \textbf{Homotopy Invariance} \cite[Chapter 4~Proposition 2.4]{outerelo2009mapping}: Let $f_1:\cl{U}\lra\R[n]$ and $f_2:\cl{U}\lra\R[n]$ be two continuous maps. \correction{Suppose} $f_1 \And f_2$ are nonvanishingly homotopic on $\bd{U}$ then $$\DEG{f_1}{U} = \DEG{f_2}{U}.$$
    \label{deg_HI}
\end{enumerate}
\ref{deg_HI} plays an important role in computing the degree of arbitrary maps. The standard approach is as follows: for a given $f$, find $f_1$ such that $\DEG{f_1}{U}$ is known apriori or
easily computable, and $f$ and $f_1$ are nonvanishingly homotopic on $\bd{U}$. Use \ref{deg_NI} along with \ref{deg_HI} to comment on the existence of the equilibrium points of $f$ on $U$. 

We now state two key results for two functions to be nonvanishingly homotopic; these results will be useful in proving the main contributions of this article.
\begin{thm}\emph{(Poincar\'e-Bohl) \cite[Chapter~2, Theorem~2.1]{krasnoselskij1984geometrical}}  
\label{pbh}
Let $U\subset\R[n]$ \correction{be open and bounded} and $f_1,f_2$ be two continuous maps on $\cl U$ with $f_1(x)\neq0$, and $f_2(x)\neq0$ for all $x\in\bd{U}$. Suppose, for no $x\in \bd {U}$, $f_1(x)$ and $f_2(x)$ are \emph{anti-parallel}, i.e., 
$$\left\langle \frac{f_1(x)}{\norm{f_1(x)}},\frac{f_2(x)}{\norm{f_2(x)}}\right\rangle\neq -1\quad \forall~x\in \bd U.$$ Then $f_1$ and $f_2$ are nonvanishingly homotopic on $\bd U$ and the underlying homotopy is 
\begin{equation*}
    [0,1]\times \bd U\ni(t,x)\mapsto H(t,x)\coloneqq tf_1(x)+(1-t)f_2(x)\in\R[n].
\end{equation*}
Moreover, one has
\begin{equation*}
    \DEG{f_1}{U} = \DEG{f_2}{U}.
\end{equation*}
\end{thm}
\begin{cor}\emph{\cite[Chapter~2, Theorem 2.3]{krasnoselskij1984geometrical}}
\label{cor_pbh}
Let $U\subset\R[n]$ be \correction{open and bounded} and $f_1:U\lra\R[n],~f_2:U\lra\R[n]$ be two continuous maps satisfying 
\begin{equation}
\label{eq_cor_pbh}
    \norm{f_1(x)-f_2(x)}<\norm{f_1(x)} \quad \forall~x\in \bd U.
\end{equation}
Then $f_1$ and $f_2$ are nonvanishingly homotopic on $\bd U$ with 
\begin{equation*}
    \DEG{f_1}{U} = \DEG{f_2}{U}.
\end{equation*}
\correction{Moreover, if $0\notin f_1(\bd U)$ and $0\notin f_2(\bd U)$, then one can replace \eqref{eq_cor_pbh} with $$\norm{f_1(x)-f_2(x)}\leq \norm{f_1(x)}\quad\forall~x\in\bd U,$$ and the assertion still holds.}
\end{cor}
\section{Estimates for the IFT and ImFT}
\label{key_res}
The ImFT and IFT find applications in proving several existential results in mathematical analysis and also provide the basis for several engineering algorithms. These theorems have a rich history and have been studied with various degree of generalizations by several authors in various sources \cite{zeidler1993vol,zeidler2012applied1,spivak2018calculus,clarke1976inverse,lang2013undergraduate,krantz2002implicit,abraham2012manifolds}. We provide prototypical versions of the ImFT and IFT below:
\begin{thm} \emph{(Implicit  Function Theorem) \cite[Theorem~4.B]{zeidler1993vol}}
\label{imft1}
Let $U\subset\R[n]$ and $V\subset\R[m]$ be open and $U\times V\ni(x,y)\mapsto f(x,y)\in\R[m]$ be a \correction{$\cont{\nu}$ map where $\nu\geq1$}. Suppose $(x_0,y_0)\in U\times V$ is such that $\pdf[]{f}{y}(x_0,y_0)\: \R[m]\lra\R[m]$ is an isomorphism. Then for any neighborhood  $\Open{y_0}$ of $y_0$, there is a neighborhood $\Open{x_0}$ of $x_0$ and a \correction{$\cont{\nu}$} map $g\: \Open{x_0}\lra \Open{y_0}$ satisfying $f(x,g(x))=w_0\coloneqq f(x_0,y_0)$ for all $x\in \Open{x_0}$. \correction{ Furthermore, we have 
\begin{equation*}
 \D g(x) = -(\pdf{f}{y}(x,g(x))^{-1}\pdf{f}{x}(x,g(x)). 
\end{equation*}
}
\end{thm}
\begin{thm} \emph{(Inverse Function Theorem) \cite[Theorem~2.5.2]{abraham2012manifolds}} 
\label{ift}
Let $U\subset\R[n]$ be a nonempty open set. Let $U\ni x\mapsto f(x)\in\R[n]$ be a \correction{$\cont{\nu}$ map with $\nu\geq1$}. Suppose $x_0\in U$ is such that $\D f(x_0)\: \R[n]\lra\R[n]$ is an isomorphism. Then there exists a neighborhood $\openn$ of $x_0$ such that $f\restrto{\openn}$ is a \correction{$\cont{\nu}$ \emph{diffeomorphism}} to its image. Moreover, for $(x,y)\in \openn\times f(\openn)$ satisfying $y=f(x)$, one has
\begin{equation*}
 \D f^{-1}(y) = (\D f(x))^{-1}.   
\end{equation*}
\end{thm}
Theorem \ref{imft1} and \ref{ift} \correction{contain assertions about the existence} of certain neighborhoods, but they do not give any information about the sizes of these neighborhoods. Our objective is to arrive at a lower bound on the sizes of these neighborhoods. For this, we use results from degree theory cataloged in Section \ref{deg_theory}. The following \correction{Theorem} provides one set of lower \correction{bounds on the sizes of} $\Open{x_0}$ and $\Open{y_0}$ defined in Theorem \ref{imft1}.
\begin{thm}
\label{imft_est}
Let $U\subset\R[n]$ and $V\subset\R[m]$ be open and $U\times V\ni (x,y)\mapsto f(x,y)\in\R[m]$ be as in Theorem \ref{imft1} and \correction{in addition let $f$ be $\cont{2}$}. Define $$M_y \coloneqq \norm{(\pdf{f}{y}(x_0,y_0))^{-1}},\And L_x \coloneqq \norm{\pdf{f}{x}(x_0,y_0)}.$$ For given $R_x>0~,\And R_y>0$ \textcolor{blue}{such that $\ball{x_0}{R_x}\subset U$ and $\ball{y_0}{R_y}\subset V$} define
\begin{equation*}
\begin{split}
    K_{xx}&\coloneqq\supr{\correction{\norm{\pdf[2]{f}{x}(x,y)}}}{(x,y)\in\ball{x_0}{R_x}\times\ball{y_0}{R_y}},\\
    K_{xy}&\coloneqq\supr{\correction{\norm{\pdf{f}{x,y}(x,y)}}}{(x,y)\in\ball{x_0}{R_x}\times\ball{y_0}{R_y}}, \And\\
    K_{yy}&\coloneqq\supr{\correction{\norm{\pdf[2]{f}{y}(x,y)}}}{(x,y)\in\ball{x_0}{R_x}\times\ball{y_0}{R_y}}.\\
\end{split}    
\end{equation*}
Then for all $0<\rx<R_x$ and $0<\ry<R_y$ satisfying
\begin{subequations}\label{eqn:epsbnd}
      \begin{align}
       \correction{\frac{1}{2}K_{xx}\rx^2+K_{xy}\rx\ry+\frac{1}{2}K_{yy}\ry^2}&<\correction{\frac{\ry}{M_y}-\rx L_x}\qand \label{subeqn-1:epsbnd} \\
        K_{xy}\rx+K_{yy}\ry&<\frac{1}{M_y}, \label{subeqn-2:epsbnd} 
    \end{align}
    \end{subequations}
for each $x\in\ball{x_0}{\rx}$, there exists a unique $y_x\in\ball{y_0}{\ry}$ satisfying $f(x,y_x)=w_0$. \correction{Furthermore, the map $\ball{x_0}{\rx}\ni x\mapsto y_x \eqqcolon g(x)\in\ball{y_0}{\ry}$ is $\cont{2}$ and $$\D g(x) = -(\pdf{f}{y}(x,y_x))^{-1}\pdf{f}{x}(x,y_x).$$}  
\end{thm}
\correction{Before we prove Theorem \ref{imft_est}, we state the finite-dimensional version of Lemma 2.5.4  from \cite{abraham2012manifolds} which shall be of use in proving Theorem \ref{imft_est} and several results to follow.}
\begin{lem}
\emph{\cite[Lemma~2.5.4]{abraham2012manifolds}} 
\label{lem_mat_inverse}
\correction{
Let $\mathrm{GL}(n,\R)$ be the set of all linear isomorphisms on $\R[n]$. Then $\mathrm{GL}(n,\R)$ is open. Moreover, for any $M\in \mathrm{GL}(n,\R)$ and a linear map $B:E\to F$ with $\norm{B}<1/\norm{M}$, $M+B\in\mathrm{GL}(n,\R)$.}  
\end{lem}
\begin{proof}[Proof of Theorem~{\upshape\ref{imft_est}}]
\correction{Assume $K_{xx}$, $K_{xy}$ and $L_x$ are not all zero.}  Let $\rx$ and $\ry$ satisfy \eqref{eqn:epsbnd}. Fix an arbitrary $x\in\ball{x_0}{\rx}$. Our objective here is to show the existence of a unique $y_x\in\ball{y_0}{\ry}$ satisfying $f(x,y_x) =w_0.$ Keeping $x$ fixed, we define  
\begin{equation}
\label{Fx}
V\ni y \mapsto \F{x}{y}\coloneqq f(x,y)-w_0\in\R[m].
\end{equation}
The problem of finding some $y_x$ satisfying $f(x,y_x) = w_0$ is now equivalent \correction{to} finding $y_x$ such that $\F{x}{y_x}=0$. The proof is set into two parts.
\\ \\
\textbf{Existence} : For ease of representation we define $\Tilde{x} \coloneqq x-x_0$ and $\Tilde{y} \coloneqq y-y_0$. Define an affine approximation of $F$ about $(x_0,y_0)$ as  
\begin{equation}
\label{Fx1}
  y\mapsto \Fapx{x}{y}\coloneqq  f(x_0,y_0)+\pdf[]{f}{x}(x_0,y_0)\Tilde{x}
+\pdf[]{f}{y}(x_0,y_0)\Tilde{y}-w_0\in\R[m].  
\end{equation}
\textit{Claim 1: $\F{x}{\cdot}$ and $\Fapx{x}{\cdot}$ are nonvanishingly homotopic on $\bd{\ball{y_0}{\ry}}$} : 
From the definition of $\Fapx{x}{\cdot}$ we have 
\begin{equation*}
\begin{split}
\F{x}{y}-\Fapx{x}{y} &= f(x,y)-\big(f(x_0,y_0)+\pdf{f}{x}(x_0,y_0)\Tilde{x}+\pdf{f}{y}(x_0,y_0)\Tilde{y}\big)\\
&=(f(x,y)-f(x_0,y_0))-\D f(x_0,y_0)\cdot\pmat{\Tilde{x}\\ \Tilde{y}}
\end{split}
\end{equation*}
Using the fundamental theorem of calculus \cite[Theorem~6.24]{rudin1964principles} we write
\begin{equation*}
\begin{split}
\F{x}{y}-\Fapx{x}{y} &=\int_0^1 \D f(x_0+s\Tilde{x},y_0+s\Tilde{y})\cdot\pmat{\Tilde{x}\\\Tilde{y}}\d s - \D f(x_0,y_0)\cdot\pmat{\Tilde{x}\\ \Tilde{y}}\\
&=\int_{0}^{1}\int_{0}^{1} \D^2f(x_0+st\Tilde{x},y_0+st\Tilde{y})\cdot((s\Tilde{x},s\Tilde{y}),(\Tilde{x},\Tilde{y}))\d t \d s. 
\end{split}
\end{equation*}
From the bounds on $\pdf[2]{f}{x},~\pdf[]{f}{x,y}$, and $\pdf[2]{f}{y}$, we have 
\begin{equation*}
\norm{\F{x}{y}-\Fapx{x}{y}}\leq \frac{1}{2}K_{xx}\norm{\Tilde{x}}^2 + K_{xy}\norm{\Tilde{x}}\norm{\Tilde{y}} +\frac{1}{2}K_{yy}\norm{\Tilde{y}}^2.
\end{equation*}
Since $x\in\ball{x_0}{\rx}$ we have $\norm{\Tilde{x}}<\rx$ and for all $y\in\bd{\ball{y_0}{\ry}}$ we have
\begin{equation}
\label{bnd1_fg}
\norm{\F{x}{y}-\Fapx{x}{y}}<\frac{1}{2}K_{xx}\rx^2 + K_{xy}\rx\ry +\frac{1}{2}K_{yy}\ry^2.
\end{equation}
Moreover, 
\begin{equation}
\label{bnd2_fg}
\begin{split}
\norm{\Fapx{x}{y}} &= \norm{f(x_0,y_0)+\pdf{f}{x}(x_0,y_0)\Tilde{x}+\pdf{f}{y}(x_0,y_0)\
\Tilde{y}-w_0}\\
&=\norm{\pdf{f}{x}(x_0,y_0)\Tilde{x}+\pdf{f}{y}(x_0,y_0)\
\Tilde{y}}\\
&\geq \left\lvert\norm{\pdf{f}{y}(x_0,y_0){\Tilde{y}}} - \norm{\pdf{f}{x}(x_0,y_0){\Tilde{x}}}\right\lvert\\
&\geq \frac{1}{M_y}\norm{\Tilde{y}}-L_x \norm{\Tilde{x}}\\
& >\frac{1}{M_y}\ry-L_x\rx ~ \forall y\in\bd\ball{y_0}{\ry}.
\end{split}
\end{equation}
\textcolor{blue}{We restate \eqref{subeqn-1:epsbnd} for convenience. 
\begin{equation*}
    \frac{1}{2}K_{xx}\rx^2 + K_{xy}\rx\ry +\frac{1}{2}K_{yy}\ry^2<\frac{\ry}{M_y}-L_x\rx.
\end{equation*}}
Thus, from \eqref{subeqn-1:epsbnd}, \eqref{bnd1_fg},  and \eqref{bnd2_fg} we have
\begin{equation}
    \norm{\F{x}{y}-\Fapx{x}{y}}\textcolor{blue}{<}\norm{\Fapx{x}{y}} ~ \forall y\in\bd\ball{y_0}{\ry},
\end{equation}
and therefore from Corollary \ref{cor_pbh} it follows that $\F{x}{\cdot}$ is nonvanishingly homotopic to $\Fapx{x}{\cdot}$ on $\bd\ball{y_0}{\ry}$ with
\begin{equation}
\label{deg_fg}
    \DEG{\F{x}{\cdot}}{\ball{y_0}{\ry}} = \DEG{\Fapx{x}{\cdot}}{\ball{y_0}{\ry}}.
\end{equation}
\textit{Claim 2: $\DEG{\Fapx{x}{\cdot}}{\ball{y_0}{\ry}} \in \{-1,1\}$} : Let $y^*\in\R[m]$ be such that $$\Fapx{x}{y^*} = 0$$ (existence and uniqueness of such a $y^*$ is guaranteed by the fact that $\Fapx{x}{\cdot}$ is an affine map in $y$). Then we have
\begin{equation*}
    y^*-y_0 = -\pdf{f}{y}(x_0,y_0)^{-1}\pdf{f}{x}(x_0,y_0)(x-x_0).
\end{equation*}
Taking norms\correction{, we get}
\begin{equation*}
\begin{split}
\norm{y^*-y_0} &= \norm{-\pdf{f}{y}(x_0,y_0)^{-1}\pdf{f}{x}(x_0,y_0)\Tilde{x}}\\
&\leq \norm{\pdf{f}{y}(x_0,y_0)^{-1}}\norm{\pdf{f}{x}(x_0,y_0)}\norm{\Tilde{x}}\\
&\leq M_yL_x\norm{\Tilde{x}}\\
& \leq M_yL_x\rx.
\end{split}
\end{equation*}
Rewriting \eqref{subeqn-1:epsbnd} as
\begin{equation*}
    \ry>M_yL_x\rx+ M_y\left(\frac{1}{2}K_{\textcolor{blue}{xx}}\rx^2+K_{xy}\rx\ry+\frac{1}{2}K_{yy}\ry^2\right),
\end{equation*}
we have 
\begin{equation}
    \norm{y^*-y_0}\leq M_yL_x\rx<\ry,
\end{equation}
which implies $y^*\in\ball{y_0}{\ry}$ and therefore 
\begin{equation}
\label{deg_fg_apx}
\begin{split}
\DEG{\Fapx{x}{\cdot}}{\ball{y_0}{\ry}} &= \sign(\det({\D\Fapx{x}{y^*}})) \\
&= \sign(\det({\pdf{f}{y}(x_0,y_0)}))\in\{-1,1\}.
\end{split}
\end{equation}
From \eqref{deg_fg} we  know $\DEG{\F{x}{\cdot}}{\textcolor{blue}{\ball{y_0}{\ry}}} \in\{-1,1\} \neq 0.$ \ref{deg_NI} now asserts the existence of a $y_x\in\ball{y_0}{\ry}$ satisfying
\begin{equation*}
    \F{x}{y_x} \coloneqq f(x,y_x) = 0.
\end{equation*}
\textbf{Uniqueness} : To complete the proof, all that remains is to show the uniqueness of $y_x$. 
\\ \\
\textit{Claim 3: $\pdf{f}{y}(x,y)$ is invertible for all $(x,y)\in\ball{x_0}{\rx}\times\ball{y_0}{\ry}$} : From the fundamental theorem of calculus \cite[Theorem~6.24]{rudin1964principles} we have
\begin{equation*}
    \pdf{f}{y}(x,y) = \pdf{f}{y}(x_0,y_0)\big(\eye{}+\pdf{f}{y}(x_0,y_0)^{-1}M(x,y)\big), 
\end{equation*}
where $$M(x,y)\coloneqq \int_{0}^{1}\left(\pdf{f}{x,y}(x_0+s\x,y_0+s\y)\x +\pdf[2]{f}{y}(x_0+s\x,y_0+s\y)\y \right)\d s.$$ Moreover, 
\begin{equation*}
\begin{split}
\norm{M(x,y)}&=\norm{\int_{0}^{1}\left(\pdf{f}{x,y}(x_0+s\x,y_0+s\y)\x+\pdf[2]{f}{y}(x_0+s\x,y_0+s\y)\y\right)\d s}\\
&\leq K_{xy}\norm{\x}+K_{yy}\norm{\y}\\
&< K_{xy}\rx+K_{yy}\ry~ \forall (x,y)\in\ball{x_0}{\rx}\times\ball{y_0}{\ry}.    
\end{split}    
\end{equation*}
From \eqref{subeqn-2:epsbnd} we have
\begin{equation*}
    \norm{M(x,y)}<K_{xy}\rx+K_{yy}\ry\correction{<} \frac{1}{M_y} ~ \forall (x,y)\in\ball{x_0}{\rx}\times\ball{y_0}{\ry}.
\end{equation*}
Hence, using Lemma \ref{lem_mat_inverse}, we see that $\pdf{f}{y}(x,y)$ is invertible for all $x,y\in\ball{x_0}{\rx}\times\ball{y_0}{\ry}.$
\\ \\
\correction{
\textit{Claim 4: $\F{x}{\cdot}$ has finite number of zeros in $\ball{y_0}{\ry}$.} : From definition of $\F{x}{\cdot}$ we have
\begin{equation*}
    \D \F{x}{y} = \pdf{f}{y}(x,y).
\end{equation*}
Since $\pdf{f}{y}(x,y)$ is invertible for all $(x,y)\in\ball{x_0}{\rx}\times\ball{y_0}{\ry}$, $0$ is a regular value of $\F{x}{\cdot}\restrto{\ball{y_0}{\ry}}$. Therefore, using Theorem \ref{ift} (Inverse Function Theorem) $F^{-1}(0~;x) \coloneqq\{y\in\ball{y_0}{\ry}\mid \F{x}{y}=0\}$ is discrete i.e. around each $y_x\in\ball{y_0}{\ry}$ there exists a small enough neighborhood such that there is no other $y$ satisfying $\F{y}{x} = 0$. From the continuity of $\F{x}{\cdot}$, $F^{-1}(0~;x)$ is closed and is bounded since $F^{-1}(0~;x)\subset\cl\ball{y_0}{\ry}$ and therefore compact and hence finite. 
\\ \\ 
Let $\F{x}{\cdot}$ have $N$ such zeros in $\ball{y_0}{\ry}$. Then from Definition \ref{deg_f} and the fact that $\pdf{f}{y}(x,y)$ is invertible on $\ball{x_0}{\rx}\times\ball{y_0}{\ry}$ we have 
\begin{equation*}
    \pm1 = \DEG{\F{x}{\cdot}}{\ball{y_0}{\ry}} = \pm N
\end{equation*}
and thus $N=1$. Therefore for each $x\in\ball{x_0}{\rx}$ there exists a unique $y_x\in\ball{y_0}{\ry}$ such that $f(x,y_x)=0$. 
\\ \\
The regularity of the map $x\mapsto g(x)\coloneqq y_x$ is a consequence of Theorem \ref{imft1} (Implicit function theorem).
\\ \\
\correction{Suppose $K_{xx} = K_{xy}=L_x=0$} then, for each $y\in V$, $f(\cdot,y)$ is a constant map and $f(x,y_0) =0$ for all $x$. Since $f(x,\cdot) = f(x_0,\cdot)$ is now a map from $\R[m]\lra\R[m]$, the uniqueness of $y_0$ can be established using estimates on the IFT which we prove in the following corollary.
} 
\end{proof}
\begin{cor}
\label{ift_est}
Let $U\subset\R[n]$ be nonempty and open and $U\ni x\mapsto f(x)\in \R[n]$ be as in Theorem \ref{ift} and \correction{in addition let $f$ be $\cont{2}$}. Define $L\coloneqq \norm{\D f(x_0)}$ and $M\coloneqq\norm{\D f^{-1}(y_0)}$. For a given $R>0$ set $K \coloneqq \supr{\norm{\D^2 f(x)}}{x\in\ball{x_0}{R}}$, and define $$P \coloneqq \min\left\{\frac{1}{MK},R\right\},\qand P'\coloneqq \frac{P(2-MKP)}{2M}.$$
Further, set $N=8M^3K$ and define
\begin{equation*}
Q\coloneqq \min\left\{\frac{1}{NL},\frac{P}{2M},P\right\} \qand Q'\coloneqq\frac{Q(2-LNQ)}{L}.
\end{equation*}
Then there exist
\begin{enumerate}[label=\textup{(\ref{ift_est}\alph*)},leftmargin=*, widest=b, align=left]
    \item an open set $H\subset\ball{x_0}{P}$ such that $f$ maps $H$ diffeomorphically onto $\ball{y_0}{P'}$ with $y_0\coloneqq f(x_0)$, and 
    \label{ift_est_a}
    \item an open set $H'\subset\ball{y_0}{Q}$ such that $f^{-1}$ maps $H'$ diffeomorphically onto $\ball{x_0}{Q'}$\correction{.}
    \label{ift_est_b}
\end{enumerate}
A graphical representation of these sets is shown in Figure \ref{inv_nbd}.
\end{cor}
\correction{
\begin{proof}  Define $(x,y)\mapsto\psi(x,y) \coloneqq f(x)-y$. From the definition we have $\psi$ a $\cont{2}$ map and $\psi(x_0,y_0) = f(x_0)-y_0 = 0$. Further, we have $\pdf{\psi}{x}(x_0,y_0) = \D f(x_0)$ and is nonsingular. Therefore $\psi$ satisfies the assumptions of Theorem \ref{imft1} and Theorem \ref{imft_est}. Computing quantities as given in Theorem \ref{imft_est} we have
\begin{equation*}
\begin{split}
M&\coloneqq M_x = \norm{(\pdf{\psi}{x}(x_0,y_0)^{-1}} = \norm{(\D f(x_0))^{-1}}\\
L_y&\coloneqq \norm{\pdf{\psi}{y}(x_0,y_0)} = 1.
\end{split}
\end{equation*}
and for $R\coloneqq R_x>0$ and $R_y>0$ such that $\ball{x_0}{\rx}\subset U$ and $\ball{y_0}{\ry}\subset V$ we have 
\begin{equation*}
\begin{split}
    K_{xx}&=\supr{\norm{\pdf[2]{\psi}{x}(x,y)}}{(x,y)\in\ball{x_0}{R_x}\times\ball{y_0}{R_y}},\\
    &=\supr{\norm{\D^2f(x)}}{x\in\ball{x_0}{R_x}}\eqqcolon K\\
    K_{xy}&\coloneqq\supr{\norm{\pdf{\psi}{x,y}(x,y)}}{(x,y)\in\ball{x_0}{R_x}\times\ball{y_0}{R_y}} = 0, \And\\
    K_{yy}&\coloneqq\supr{\norm{\pdf[2]{\psi}{y}(x,y)}}{(x,y)\in\ball{x_0}{R_x}\times\ball{y_0}{R_y}} =0.
\end{split} 
\end{equation*}
Simplifying equation \ref{eqn:epsbnd} for $\psi$, for each $0<\rx<1/MK$ and for $\ry= \frac{\rx(2-MK\rx)}{2M}$, for each $y\in\ball{y_0}{\ry}$ there exists an unique $x_y\in\ball{x_0}{\rx}$ satisfying $\psi(x_y,y) =0 \iff f(x_y)=y.$ Further, $y\mapsto x_y\eqqcolon f^{-1}(y)$ is $\cont{2}$. Defining $H:= f^{-1}(\ball{y_0}{P'})\cap\ball{x_0}{P}$ proves \ref{ift_est_a}.
\\ \\
For \ref{ift_est_b}, from the relation $f^{-1}\circ f = \text{identity}$, for any $u_1,u_2\in\R[n]$ we have 
\begin{equation*}
    \D^2 f^{-1}(y)(\D f(x)u_1,\D f(x)u_2)+\D f^{-1}(y)\D^2 f(x)(u_1,u_2) = 0.
\end{equation*}
After some rearrangement and imposing bounds on $\norm{\D f^{-1}(y)}$, for all $x\in\ball{x_0}{P/2M}$ we arrive at $$\norm{\D^2f^{-1}(y)}<8M^3K=N~\forall y\in\ball{y_0}{P/2M}.$$ 
Applying \ref{ift_est_a} on $f^{-1}$ gives \ref{ift_est_b}. 
\end{proof} 
}
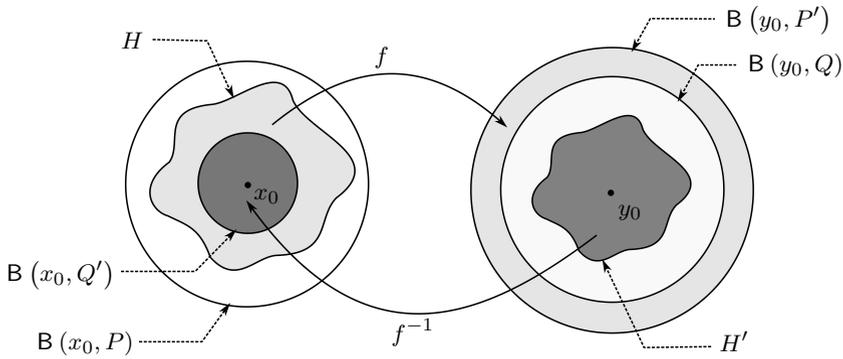
\begin{figure}
\centering
\begin{tikzpicture}
\node [
    above,
    inner sep=0] (image) at (0,0) {\includegraphics[width=0.70\linewidth]{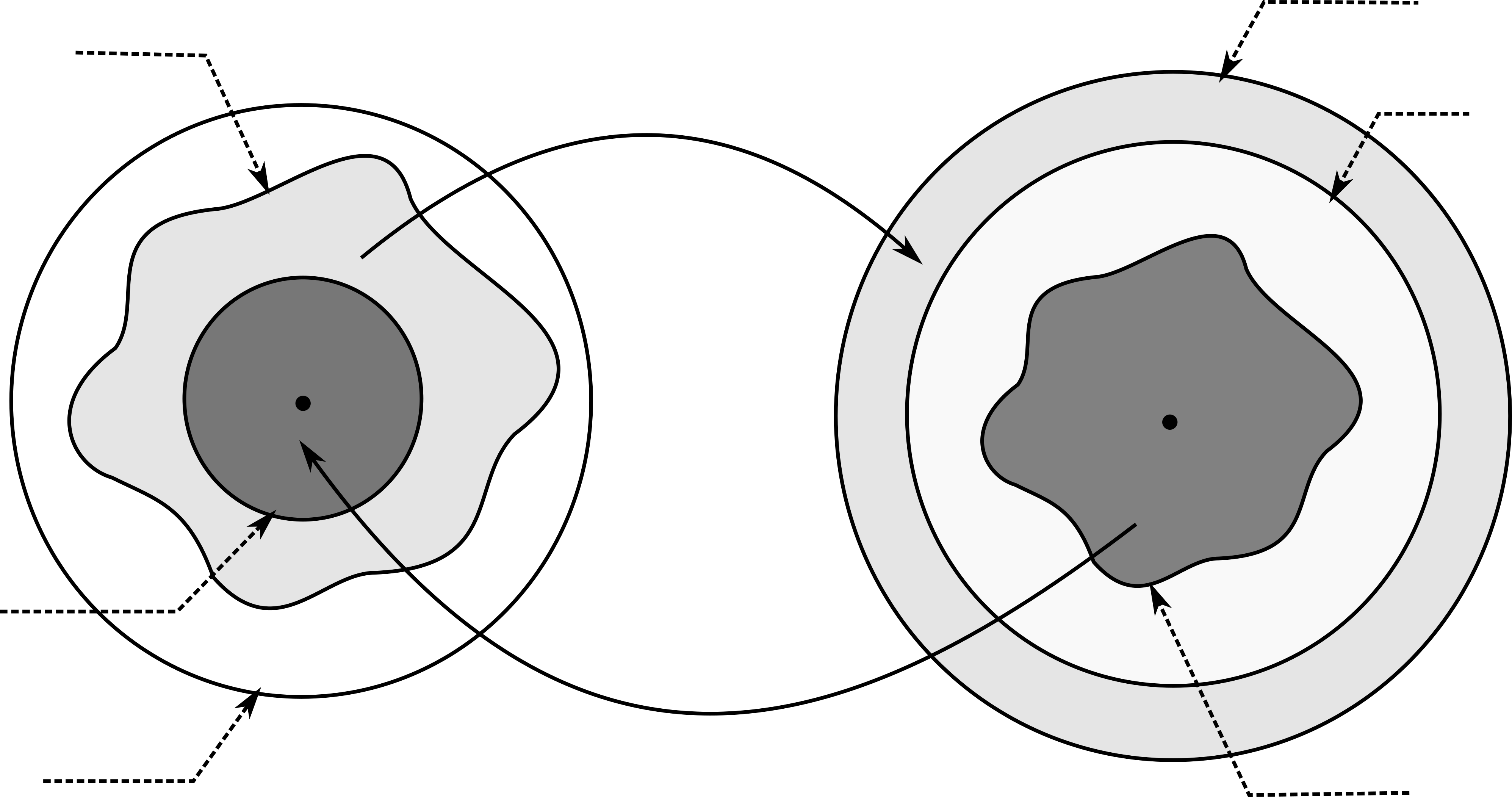}};
\begin{scope}[
x={($0.1*(image.south east)$)},
y={($0.1*(image.north west)$)}]
\draw (18,9.2) 
        node[above right,black]{\small{$\ball{y_0}{P'}$}};
\draw(-8.4,1.5) 
        node[above left,black]{\small{$\ball{x_0}{Q'}$}};
\draw(-9.8,-0.5) 
        node[above left,black]{\small{$\ball{x_0}{P}$}};
\draw(17.4,7.9) 
        node[above right,black]{\small{$\ball{y_0}{Q}$}};
\draw(7.1,9.4) 
        node[below right,black]{\small{$f$}};
\draw(0,-0.2) 
        node[above left,black]{\small{$f^{-1}$}};
\draw(-2,4.1) 
        node[above right,black]{\small{$x_0$}};
\draw(10,4.6) 
        node[below right,black]{\small{$y_0$}};
\draw(9.4,0.8) 
        node[below right,black]{\small{$H'$}};
\draw(-1.5,8.8) 
        node[above right,black]{\small{$H$}};
\end{scope}
\end{tikzpicture}\\
\caption{\centering Neighborhoods given in Corollary \ref{ift_est}}
    \label{inv_nbd}
\end{figure}
\begin{rmk}
\label{ift_rem}
Here are some remarks on  Theorem \ref{imft_est} and Corollary \ref{ift_est}.
\begin{enumerate}[label=\textup{(\ref{ift_rem}\alph*)},leftmargin=*, widest=b, align=left]
    \item Equation \eqref{subeqn-2:epsbnd} ensures uniqueness of $y_x$. If one is only interested in showing the existence of $y_x$, then one may relax \eqref{eqn:epsbnd} to just \eqref{subeqn-1:epsbnd}. 
    \item For finite-dimensional spaces, the bounds given by \correction{Corollary \ref{ift_est}} are better than those given by Abraham et al.~ \cite[Proposition~2.5.6]{abraham2012manifolds}. : \correction{ Proposition 2.5.6 from \cite{abraham2012manifolds} provides the bounds for the estimates for the IFT. Bounds given by Corrolary \ref{ift_est} for the IFT are applicable on a larger domain than that of \cite[Proposition~2.5.6]{abraham2012manifolds} and hence are superior.} 
    \item Since our objective is to come up with bounds that require minimal numerical computation, we relax several inequalities on the way and use conservative versions of them. For instance, we use Corollary \ref{cor_pbh} which is a coarser version of Theorem \ref{pbh}. Further, in several places, we utilize the submultiplicativity of norms to arrive at these results.
    \item The proof for Theorem \ref{imft_est} relies on constructing an approximation $\Fapx{x}{\cdot}$ of the map $\F{x}{\cdot}$ and showing both of them are nonvanishingly homoptopic on $\bd\ball{y_0}{\ry}$. One may improve the obtained estimates by approximating $\F{x}{\cdot}$ with a different function $\tilde F(\cdot~;x)$ and constructing homotopy $H'$ such that $F$ and $\tilde F$ are nonvanishingly homotopic and then use analysis on $\tilde F$ to show the existence of solutions of $F(y;x)= 0$. We refer the reader to \cite[Chapter~1]{krasnoselskij1984geometrical}, in which Krasnosel'skij et al. have discussed several methods to come up with such homotopies.   
\end{enumerate}
\end{rmk}
In order to compute these estimates one needs to compute the second-order derivatives over a bounded region. In light of this, $f\in\cont{2}$ is a necessary requirement. However, for $\cont{1}$ maps, one can alternatively use the bounds on the first-order derivatives to come up with similar estimates. \correction{In \cite{holtzman1970explicit}, Holtzman provides explicit estimates for the Implicit function theorem (given by the following proposition). However, for finite-dimensional spaces, one can also obtain these estimates by degree theoretic ideas. We now present these estimates and a degree theoretic proof of the same.}
\begin{prop}\emph{\cite[Theorem]{holtzman1970explicit}}
\correction{
\label{Imft_holtzman} 
Let $\R[n]\times\R[m]\supset\Omega\ni(x,y)\mapsto P(x,y)\in\R[m]$ be a continuous map. Let $\Omega$ be open and $(x_0,y_0)\in \Omega$. Assume:
\begin{enumerate}[label=\textup{(\roman*)},leftmargin=*, widest=b, align=left]
\item $P(x_0,y_0)=0$;
\label{item1}
\item $\pdf{P}{y}$ exists and is continuous in $\Omega$;
\label{itme2}
\item $\pdf{P}{y}(x_0,y_0)$ has bounded inverse $\Gamma = (\pdf{P}{y}(x_0,y_0))^{-1}$ and $\norm{\Gamma} = k_1$;
\label{item3}
\item $S = \{(x,y)\in\ball{x_0}{\delta}\times\cl\ball{y_0}{\epsilon}\}\subset\Omega$;
\label{item4}
\item there is a real-valued function $g_1(u,v)$ defined for $(u,v)\in[0,\delta]\times[0,\epsilon]$ and nondecreasing in each argument with the other fixed such that for all $(x,y)\in S$
\begin{equation*}
    \norm{\pdf{P}{y}(x,y)-\pdf{P}{y}(x_0,y_0)}\leq g_1(\norm{x-x_0},\norm{y-y_0});
\end{equation*}
\label{item5}
\item there is a nondecreasing function $g_2$ defined on $[0,\delta]$ such that for all $(x,y)\in S$ 
\begin{equation*}
    \norm{P(x,y_0)}\leq g_2(\norm{x-x_0});
\end{equation*}
\label{item6}
\item $k_1g_1(\delta,\epsilon)\leq\alpha<1$, and $k_1g_2(\delta)\leq \epsilon(1-\alpha)$.
\label{item7}
\end{enumerate}
Then a unique map $F$ exists defined on $\ball{x_0}{\delta}$ mapping into $\cl\ball{y_0}{\epsilon}$ and possessing the following properties:
\begin{enumerate}[label=\textup{(\alph*)},leftmargin=*, widest=b, align=left]
    \item $P(x,F(x)) = 0$ for all $x\in\ball{x_0}{\delta}$.
    \label{itema}
    \item $F(x_0) = y_0$.
    \label{itemb}
    \item $x\mapsto F(x)$ is continuous.
    \label{itemc}
\end{enumerate}}
\end{prop}
\begin{proof}
\correction{We first prove the strict version of \ref{item7} and then using this we achieve the nonstrict version \ref{item7}. To this end, let $\epsilon>0$ and $\delta>0$ be such that $k_1g_1(\delta,\epsilon)<\alpha<1$ and $k_1g_2<\epsilon(1-\alpha)$ and $\ball{x_0}{\delta}\times \cl\ball{y_0}{\epsilon}\subset\Omega$. Fix an $x\in\ball{x_0}{\delta}$ and define $\cl\ball{y_0}{\epsilon}\ni y\mapsto \P{x}{y}\coloneqq P(x,y)-P(x_0,y_0)$. Further, define an approximation of $\P{x}{\cdot}$ as 
\begin{equation}
\label{e1.3}
y\mapsto \Papx{x}{y} \coloneqq P(x,y_0)+ \pdf{P}{y}(x_0,y_0)\Tilde{y}    
\end{equation}
where $\Tilde{y} \coloneqq y-y_0$ and $\Tilde{x} \coloneqq x-x_0$. From definition of $\P{x}{\cdot}$ we have 
\begin{equation*}
\begin{split}
\P{x}{y}- \Papx{x}{y} &= (P(x,y)-P(x,y_0)) - \pdf{P}{y}(x_0,y_0)\Tilde{y}.\\
&=\int_{0}^{1}\Big(\pdf{P}{y}(x,y_0+t\Tilde{y})-\pdf{P}{y}(x,y_0)\Big)\tilde{y}\d t, 
\end{split}    
\end{equation*}
and 
\begin{equation*}
    \norm{\P{x}{y}- \Papx{x}{y}}\leq \left(\int_{0}^{1}\norm{\pdf{P}{y}(x,y_0+t\Tilde{y})-\pdf{P}{y}(x,y_0)}\d t\right)\norm{\tilde{y}}. 
\end{equation*}
Therefore, from \ref{item5}, for all $y\in\bd\ball{y_0}{\epsilon}$ 
\begin{equation}
\label{eq3.13}
    \norm{\P{x}{y}- \Papx{x}{y}}\leq g_1(\delta,\epsilon)\epsilon \leq \frac{\alpha\epsilon}{k_1}.
\end{equation}
Similarly, using \ref{item3}, \ref{item6} and \eqref{e1.3}, for all $y\in\bd\ball{y_0}{\epsilon}$ we have 
\begin{equation}
\label{eq3.14}
    \norm{\Papx{x}{y}}\geq \frac{\epsilon}{k_1} -g_2(\delta).
\end{equation} 
Since by assumption $k_1g_2(\delta)<\epsilon(1-\alpha)$, from \eqref{eq3.13} and \eqref{eq3.14} we have 
\begin{equation*}
    \norm{\P{x}{y}- \Papx{x}{y}}<\norm{\Papx{x}{y}}. 
\end{equation*}
Thus using Corollary \ref{cor_pbh} we have 
\begin{equation*}
    \DEG{\P{x}{\cdot}}{\ball{y_0}{\epsilon}} = \DEG{\Papx{x}{\cdot}}{\ball{y_0}{\epsilon}} \in\{-1,1\}.
\end{equation*}
Therefore, there exists a $y_x\in\ball{y_0}{\epsilon}$ such that $P(x,y_x)=0$. 
\\ \\
Following arguments similar to those of Theorem \ref{imft_est}, for uniqueness, it suffices to show that $\pdf{P}{y}(x,y)$ is nonsingular for all $(x,y)\in S$. Rewriting $\pdf{P}{y}(x,y)$ as 
\begin{equation*}
    \pdf{P}{y}(x,y) = \pdf{P}{y}(x_0,y_0)\Big(\eye{n} - \Gamma(\pdf{P}{y}(x_0,y_0)-\pdf{P}{y}(x,y))\Big)
\end{equation*}
Using \ref{item7}, for all $(x,y)\in S$, 
\begin{align*}
    \norm{\Gamma\big(\pdf{P}{y}(x_0,y_0)-\pdf{P}{y}(x,y)\big)} &\leq \norm{\Gamma}\norm{\pdf{P}{y}(x_0,y_0)-\pdf{P}{y}(x,y)}\\ &\leq k_1\alpha<1.
\end{align*}
Hence using Lemma \ref{lem_mat_inverse}, $\pdf{P}{y}(x,y)$ is nonsingular for all $(x,y)\in S$. Therefore for all $\epsilon>0$, $\delta>0$ such that $k_1g_1(\delta,\epsilon)<\epsilon(1-\alpha)$, for all $x\in\ball{x_0}{\delta}$ there exists a unique $y_x\in\ball{y_0}{\epsilon}$ such that $P(x,y_x) =0$.  
\\ \\
For given $\delta>0$, define $$\epsilon^*\coloneqq\supr{\epsilon}{k_1g_1(\delta,\epsilon)\leq\alpha,~k_1g_2(\delta)\leq\epsilon(1-\alpha)}>0.$$Choose an $\alpha<\alpha'<1$, and extend $g_1$, $g_2$ if necessary, then there exists an $\epsilon'$ such that $k_1g_1(\delta,\epsilon')<\alpha'$ and $k_1g_2(\delta)<\epsilon'(1-\alpha')$ and $\ball{x_0}{\delta}\times\ball{y_0}{\epsilon'}\in\Omega$. One can see that $\epsilon'>\epsilon^*$. Set $$\epsilon''(\alpha') = \infr{\epsilon'}{k_1g_1(\delta,\epsilon')<\alpha',~k_1g_2(\delta)<\epsilon'(1-\alpha')}.$$ Then for each $x\in\ball{x_0}{\delta}$, there exist a unique $y_x\in\ball{y_0}{\epsilon''(\alpha')}$ such that $P(x,y_x) =0$. Let $\{\alpha_n\mid n\in\N\}$ be a monotone decreasing sequence with $\alpha_n\leq\alpha'$ and $\alpha_n\lra \alpha$. Suppose $g_1(\delta,\cdot)$ is continuous on $[\epsilon^*,\epsilon']$\footnote{Such an extension of $g_1(\delta,\cdot)$ exists and one possible choice is 
\begin{equation*}
g_1(\delta,r) = \begin{cases}
g_1(\delta,r)&~\quad r\leq\epsilon^*,\\
g_1(\delta,\epsilon^*)+\supr{\norm{\pdf{P}{y}(x,y)-\pdf{P}{y}(x_0,y_0)}}{S_r}&\\ - \supr{\norm{\pdf{P}{y}(x,y)-\pdf{P}{y}(x_0,y_0)}}{S}&~\quad r>\epsilon^*,
\end{cases}    
\end{equation*}
where $S_r = \ball{x_0}{\delta}\times\ball{y_0}{r}$.} then we have a sequence $\epsilon_n \coloneqq \epsilon''(\alpha_n)$, such that $\epsilon_n\lra\epsilon^*$. Furthermore, for each $x\in\ball{x_0}{\delta}$, there exists a unique $y_x$ such that $P(x,y_x)= 0$ and for each $n$, $y_x\in\ball{y_0}{\epsilon_n}$. Therefore $$y_x\in\bigcap_{n\in \N}\ball{y_0}{\epsilon_n}=\cl\ball{y_0}{\epsilon^*}.$$ Thus, for each $\delta,\epsilon$ satisfying \ref{item7}, for each  $x\in\ball{x_0}{\delta}$ there exists a unique $y_x\in\cl\ball{y_0}{\epsilon}$ satisfying $P(x,y_x) =0$. Defining $x\mapsto F(x) = y_x $ proves \ref{itema} and \ref{itemb}, while \ref{itemc} is a consequence of Theorem \ref{ift} (ImFT).}
\end{proof}
In Proposition \ref{Imft_holtzman}, in addition to the above assumptions, if $\pdf{P}{x}(x,y)$ exists and is continuous then one has the following corollary. 
\begin{cor}
\label{imftc1}
Let $U\subset \R[n]$ and $V\subset\R[m]$ be open sets and $U \times V \ni(x,y)\mapsto f(x,y)\in\R[n]$ be $\cont{1}$. Suppose for some $(x_0,y_0)\in U\times V$, $\pdf{f}{y}(x_0,y_0)$ is an isomorphism. Define $M_y\coloneqq \norm{(\pdf{f}{y}(x_0,y_0))^{-1}}$ and $L_x\coloneqq\norm{\pdf{f}{x}(x_0,y_0)}$. For any given $\rx>0$ and $\ry>0$ define
\begin{align*}
    \lipconst[x]{\rx} &\coloneqq \supr{\norm{\pdf{f}{x}(x,y_0)-\pdf{f}{x}(x_0,y_0)}}{(x,y)\in\ball{x_0}{\rx}},\And\\
    \lipconst[y]{\rx,\ry} &\coloneqq \supr{\norm{\pdf{f}{y}(x,y)-\pdf{f}{y}(x_0,y_0)}}{(x,y)\in\ball{x_0}{\rx}\times\ball{y_0}{\ry}}.
\end{align*}
Then for all $\rx>0$ and $\ry>0$ satisfying 
\begin{subequations}
\label{eq_eps_c1}
\begin{align}
\label{eq_eps_c1a}
\lipconst[x]{\rx}\rx+\lipconst[y]{\rx,\ry}\ry&<\frac{\ry}{M_y}-L_x\rx,\And\\
\label{eq_eps_c1b}
M_y\lipconst[y]{\rx,\ry}\leq\alpha<1\correction{,}
\end{align}      
\end{subequations}
for each $x\in\ball{x_0}{\rx}$ there exist a unique $y_x\in\ball{y_0}{\ry}$ satisfying $f(x,y_x)= w_0\eqqcolon f(x_0,y_0)$. Furthermore, $\ball{x_0}{\rx}\ni x\mapsto g(x)\coloneqq  y_x\in \ball{y_0}{\ry}$ is continuously differentiable.  
\end{cor}
\begin{proof}
\correction{Define $g_1(\rx,\ry) \coloneqq \lipconst[y]{\rx}{\ry}$ and $g_2(\rx) \coloneqq \lipconst[x]{\rx}$. Applying Proposition \ref{Imft_holtzman} with inequality \ref{item7} replaced with its strict version and rearranging terms provide \eqref{eq_eps_c1a} and \eqref{eq_eps_c1b}. The differentiability of $g$ is a consequence of Theorem \ref{imft1} (ImFT).}
\end{proof}
\begin{cor}
\label{iftc1}
Let $U\subset\R[n]$ be open and $U\ni x\mapsto f(x)\in\R[n]$ be a $\cont{1}$ map. Let $x_0\in U$ be such that $\D f(x_0)$ is nonsingular. Define $L\coloneqq \norm{\D f(x_0)}$ and $M\coloneqq \norm{(\D f(x_0))^{-1}}$. For each $r>0$, define
\begin{equation*}
    \lipconst{r} \coloneqq\supr{\correction{\norm{\D f(x)-\D f(x_0)}}}{x\in\ball{x_0}{r}}.
\end{equation*}
Then for all $\rx>0$ and $\ry>0$ satisfying
\begin{equation*}
    \lipconst{\rx}\correction{<} \frac{1}{M} \qand \ry = \frac{\rx(1-M\lipconst{\rx})}{M},
\end{equation*}
there exist an open set $H_{\rx}\subset\ball{x_0}{\rx}$ such that $f$ maps $H_{\rx}$ onto $\ball{y_0}{\ry}$ $\cont{1}$-diffeomorphically. Further, for each $0<r<\ry$, define 
\begin{equation*}
    \mathbf{M}(r) = \supr{\norm{\D f^{-1}(y)-\D f^{-1}(y_0)}}{y\in\ball{y_0}{r}}.
\end{equation*}
Then for each $\ry'>0$ and $\rx'>0$ satisfying
\begin{equation*}
    \mathbf{M}(\ry')\correction{<} \frac{1}{L} \qand \rx' = \frac{\ry{'}(1-L\mathbf{M}(\ry')}{L},
\end{equation*}
there exists an open set $H'_{\ry'}\subset\ball{x_0}{\ry'}$ such that $f^{-1}$ maps $H'_{\ry'}$, $\cont{1}-$ diffeomorphically onto $\ball{x_0}{\rx'}\subset H_{\rx}$.  
\end{cor}
\begin{rmk}
Estimating $\mathbf{M}$ requires explicit expression of $\D f^{-1} (y_0)$. This can be overcome by using the relation $\D f^{-1}(y) = (\D f(x))^{-1}$, where $x=f^{-1}(y)$. Let $r'<\rx$, be such that $r = r{'}(1-M\lipconst{r'})/M$, then we have 
\begin{equation*}
   \mathbf{M}(r)\leq\supr{\norm{(\D f(x))^{-1}-(\D f(x_0))^{-1}}}{x\in\ball{x_0}{r'}}.
\end{equation*}
\end{rmk}
\correction{Corollary} \ref{ift_est} and Corollary \ref{iftc1} considers an unconstrained variation around $y_0$. However, often we are interested to know how the preimage of $y$ under $f$ varies when $y$ is varied along a particular direction. A slightly improved version of Corollary \ref{ift_est} is presented in the following proposition.  
\begin{prop}
\label{dir_bound}
Let $U\subset\R[n]$ be a nonempty open set  and $f\:U\lra\R[n]$ satisfy the assumptions of Theorem \ref{ift} and $L$, $M$, $K$, $P$ and $R$ be as defined in \correction{Corollary} \ref{ift_est}. Let $\vecspc{W}\subset\R[n]$ be a subspace, 
define $$\ball[\vecspc{W}]{y_0}{r_\vecspc{W}} \coloneqq\{y\subset\R[n]\mid y-y_0\in\vecspc{W}~\text{and}~\norm{y-y_0}<r_\vecspc{W}\}$$ and  
\begin{equation*}
    M_{\vecspc{W}} \coloneqq \supr{\norm{\D f^{-1}(y_0){\y}}}{{\y}\in\ball[\vecspc{W}]{0}{1}}.
\end{equation*}
Then for any $$0\correction{<} \rx < \min\left\{\frac{1}{MK},R\right\} \qand r_\vecspc{W} = \frac{\rx(2-\rx MK)}{2M_{\vecspc{W}}},$$for each $y \in \ball[\vecspc{W}]{y_0}{r_\vecspc{W}}$, there exist  a unique $x_y\in\ball{x_0}{\rx}\subset\R[n]$ such that $f(x_y) = y$.  
\end{prop}
\begin{proof}
\correction{
First fix a $y\in\ball[\vecspc{W}]{y_0}{r_\vecspc{W}}$ and define $\F{y}{x} = (\D f{x_0})^{-1}(f(x)-y)$. We have $\F{y}{x} = 0\iff f(x)=y$. Define an approximation of $\F{y}{\cdot}$ as 
$$\Fapx{y}{x} = (x-x_0)-\D f(x_0)^{-1}(y-y_0).$$
Setting $\Tilde{x} \coloneqq x-x_0$ and $\Tilde{y} \coloneqq y-y_0$ we have 
\begin{align*}
 \F{y}{x}- \Fapx{y}{x}&= \D f(x_0)^{-1}(f(x-y))-(\Tilde{x}-\D f(x_0)^{-1}\Tilde{y})\\
 &= \D f(x_0)^{-1}\Big(f(x)-\big(f(x_0)+\D f(x_0)\Tilde{x}\big)\Big).
\end{align*}
Therefore, 
$$\norm{\F{y}{x}- \Fapx{y}{x}}<\frac{1}{2}MK\rx^2~\forall~x\in\ball{x_0}{\rx}.$$
Similarly, 
\begin{align*}
  \norm{\Fapx{y}{x}}&= \norm{\Tilde{x}-\D f(x_0)^{-1}\Tilde{y}}\\
  &\geq \left\lvert \norm{\Tilde{x}} - \norm{\D f(x_0)^{-1}\Tilde{y}}\right\lvert\\
  &> \rx - M_\vecspc{W} r_\vecspc{W}~\forall~x\in\bd\ball{\rx}.
\end{align*}
Then for all $\rx$ and $r_{\vecspc{W}}$ satisfying 
\begin{equation}
\label{eq3.131}
    \frac{1}{2}MK\rx^2 \leq \rx-M_\vecspc{W}r_\vecspc{W}
\end{equation}
we have (for all $y\in\ball[\vecspc{W}]{y_0}{r_\vecspc{W}}$) $\norm{\F{y}{x}-\Fapx{y}{x}}<\norm{\Fapx{y}{x}}$ for all $x\in\bd\ball{x_0}{\rx}$ and therefore are nonvanishingly homotopic. Therefore, 
\begin{equation*}
    \DEG{\F{y}{\cdot}}{\ball{x_0}{\rx}} = \DEG{\Fapx{y}{\cdot}}{\ball{x_0}{\rx}} \in\{-1,1\},
\end{equation*}
and thus there exists $x_y\in\ball{x_0}{\rx}$ such that $\F{y}{x_y}=0\iff f(x_y) = y$. This proves existence of such an $x_y\in\ball{x_0}{\rx}$ for each  $y\in\ball[\vecspc{W}]{y_0}{\vecspc{r_\vecspc{W}}}$. Since for all $r_x<\min\{ \frac{1}{MK},R\}$, $\D \F{y}{x} = \D f(x_0)^{-1}\D f(x)$ is invertible for all $x\in\ball{x_0}{\rx}$, the uniqueness of $x_y$ is asserted following arguments similar to Proposition \ref{ift_est}. For a given $\rx$, maximizing $r_\vecspc{W}$ while satisfying \eqref{eq3.131} results in the equality 
\begin{equation}
    r_\vecspc{W} = \frac{\rx(2-\rx MK)}{2M_{\vecspc{W}}},
\end{equation}
and thereby completing the proof.
}
\end{proof}

\section{Applications}
ImFT and IFT find their applications in proving several results in nonlinear analysis and form the basis of several results in system theory and control, such as the robustness of solutions, and the existence, and uniqueness of solutions of ordinary differential equations. Several existential results on robustness are a consequence of the ImFT. One can apply the bounds derived in this article on the ImFT and IFT to obtain quantitative variants of the aforementioned methods. In this section, we include two such applications. First, we look into the robustness of the solutions of the nonlinear equations with respect to parametric variations and particularly the Quadrartically constrained Quadratic Programs (QCQP). A second application of these bounds is in geometric control methods, where we utilize these bounds to give explicit estimates on the domain of the feedback linearization of discrete-time dynamical systems. 
\subsection{Robustness of Nonlinear Equations}
Solving nonlinear equations is challenging, and one often needs to employ numerical methods to compute a solution for a system of nonlinear equations. Estimates of the bounds in the ImFT and IFT help \correction{in proving} the existence of solutions for a system of nonlinear equations on a given set, as we now demonstrate.
\\ \\
Let $f\:\R[n]\times\R[m]\lra\R[n]$ be a $\cont{2}$ map. Let $X\subset\R[n]$, $U\subset\R[m]$ be nonempty open sets. A system of nonlinear equations in $(x,u)$ on $X\times U$ is given by 
\begin{equation}
\label{non_eq}
f(x,u)=0
\end{equation}
is \emph{solvable} if there exists an $(x_0,u_0)\in X\times U$ such that $f(x_0,u_0) =0$.
\begin{defn}
\label{def4.1}
Let \correction{$\Omega\subset X$} have a nonempty interior. \eqref{non_eq} is said to be \emph{robustly solvable} on $\Omega$, if there exists an $\ru>0$, such that for all $u\in\ball{u_0}{\ru}\subset U$, there exists a \correction{(unique)} $x_u\in\Omega$ satisfying $f(x_u,u) = 0$. Moreover, for a given $\rx$, the \correction{supremum over all such $\ru$s} is called the \emph{robustness margin} for \eqref{non_eq} over \correction{$\Omega$}. 
\end{defn}
\begin{rmk}
\correction{Definition \ref{def4.1} is similar to \cite[Definition~2.1]{dvijotham2019robust}. However, unlike \cite[Definition~2.1]{dvijotham2019robust}, $\Omega$ can be arbitrarily chosen and need not arise from linear constraints given by equation (2) in \cite{dvijotham2017solvability}.}     
\end{rmk}
\begin{thm}
Let $X\subset\R[n],~U\subset\R[m]$ be nonempty set and $f:X\times U\lra \R[n]$ be a $\cont{2}$ map. Let $(x_0,u_0)\in X\times U$ be such that $f(x_0,u_0)=0$ and $\pdf{f}{x}(x_0,u_0)$ is nonsingular. Define 
\begin{equation*}
    L_u \coloneqq \norm{\pdf{f}{u}(x_0,u_0)}\qand M_x \coloneqq \norm{\left(\pdf{f}{x}(x_0,u_0)^{-1}\right)}
\end{equation*}
and for a given $R_x>0,R_u>0$ such that $\ball{x_0}{R_x}\subset X$ and $\ball{u_0}{R_u}\subset U$ set
\begin{align*}
    K_{xx}&\coloneqq  \supr{\norm{\pdf[2]{f}{x}(x,u)}}{(x,u)\in\ball{x_0}{R_x}\times\ball{u_0}{R_u}},\\
    K_{uu}&\coloneqq  \supr{\norm{\pdf[2]{f}{x}(x,u)}}{(x,u)\in\ball{x_0}{R_x}\times\ball{u_0}{R_u}},\And \\
    K_{xu}&\coloneqq\supr{\norm{\pdf{f}{x,u} (x,u)}}{(x,u)\in\ball{x_0}{R_x}\times\ball{u_0}{R_u}}.
\end{align*}
Then for all $\rx$ and $\ru$ satisfying 
\begin{subequations}\label{eqn:epsrobbnd}
      \begin{align}
        \frac{1}{2}K_{xx}\rx^2+K_{xu}\rx\ru+\frac{1}{2}K_{uu}\ru^2&<\frac{\rx}{M_x}-L_u\ru, \label{subeqn-1:epsrobbnd} \\
        K_{xu}\ru+K_{xx}\rx&<\frac{1}{M_x}, \label{subeqn-2:epsrobbnd} \\
        0<\ru&\leq R_x,\And\label{subeqn-3:epsrobbnd} \\ 0< \rx&\leq R_y, \label{subeqn-4:epsrobbnd}
      \end{align}
\end{subequations}
\eqref{non_eq} is robustly solvable on $\ball{x_0}{\rx}$ and robustness margin is bounded below by $\ru$. 
\end{thm}
\begin{rmk}
\correction{Similar to Theorem \ref{imft_est}, \eqref{subeqn-2:epsrobbnd} is necessary for the uniqueness of the solution. If the uniqueness is not required then one can relax \eqref{eqn:epsrobbnd} to \eqref{subeqn-1:epsrobbnd}, \eqref{subeqn-3:epsrobbnd} and \eqref{subeqn-4:epsrobbnd}.} 
\end{rmk}
\noindent \correction{The proof follows directly from applying Theorem \ref{imft_est} on $f$, and is therefore omitted.} With the general nonlinear case described, we now consider the QCQPs.
\subsubsection{Robustness of Solutions of QCQPs}
Let $u\in\R[n]$ and define 
\begin{equation*}
\R[n]\ni x\mapsto  F(x) \coloneqq Q(x)+Lx\in\R[n],
\end{equation*}
where $L\in\R[{n\times n}]$ and $Q(x)$ is a quadratic function with its $i^\text{th}$ component defined as 
\begin{equation}
    [Q(x)]_i \coloneqq x^\top Q_i x ~\text{for all}~i\in \n.
\end{equation}
where for all $i\in \n$,~$Q_i \in \R[{n\times n}]$ is a symmetric matrix. \correction{For a given matrix $A$ and a vector $b$}, a QCQP is defined by a system of quadratic equations 
\begin{equation}
\label{quadeq}
    F(x) = Q(x)+Lx = u,
\end{equation}
subjected to $m$ constraints
\begin{equation}
\label{eq4.5}
    x\in \const\coloneqq\{x\in \R[n]~\mid~[Ax]_i \leq b_i ~i\in [m]\},
\end{equation}
where \correction{$[Ax]_i$ and $b_i$ denotes the $i^\text{th}$ components of the vectors $Ax$ and $b$.} The constraint set $\const$ defines a \correction{polyhedron} as shown in Figure \ref{constset}. The dashed lines represent the equalities $[Ax]_i=b_i$ and the shaded interior along with the boundary defines the constraint set. 
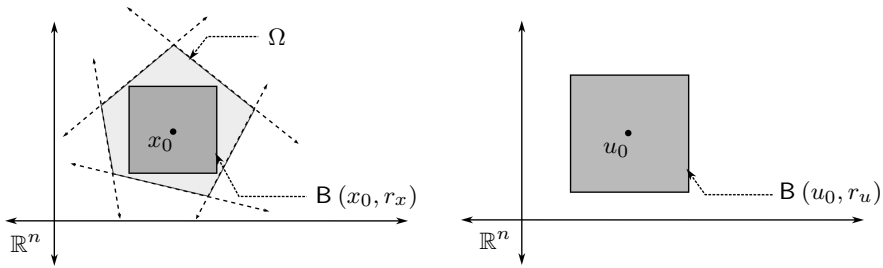
\begin{figure}
\begin{center}
\begin{tikzpicture}
\node [
    above,
    inner sep=0] (image) at (0,0) {\includegraphics[width=0.95\linewidth]{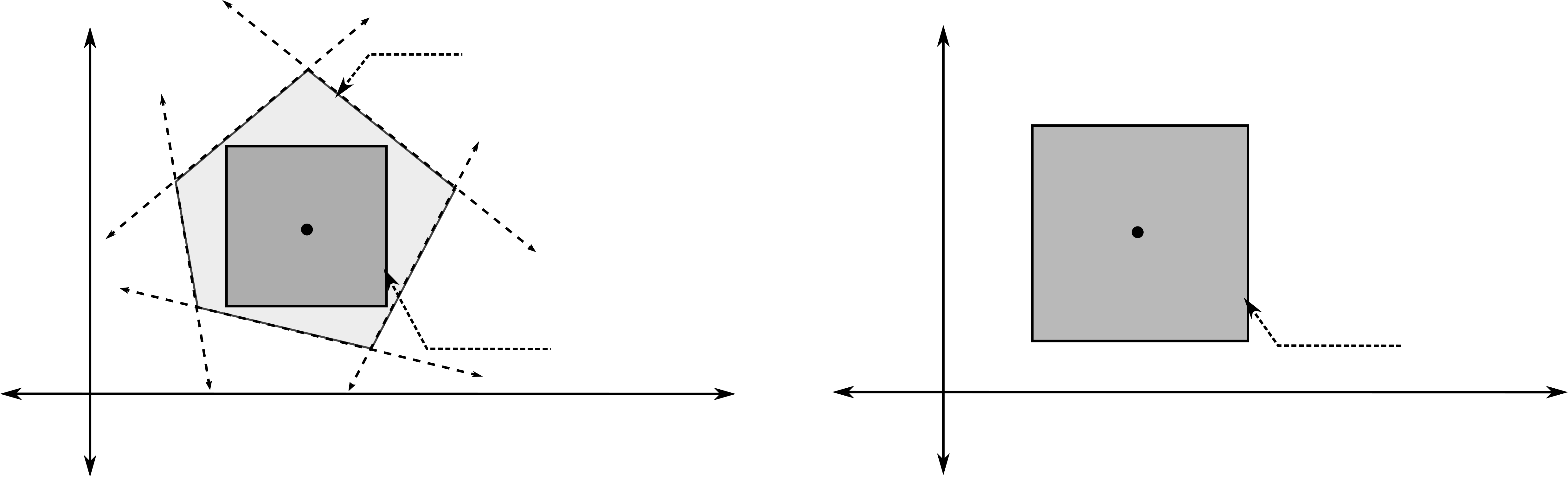}};
\begin{scope}[
x={($0.1*(image.south east)$)},
y={($0.1*(image.north west)$)}]
\draw(-9.9,0.2) 
        node[above right,black]{\small{$\R[n]$}};
\draw(5.4,9.5) 
        node[below right,black]{\small{$\const$}};
\draw(2.3,0.3) 
        node[above left,black]{\small{$\R[n]$}};
\draw(-2.8,4.1) 
        node[above right,black]{\small{$x_0$}};
\draw(8.8,5.1) 
        node[below right,black]{\small{$u_0$}};
\draw(0.5,3.5) 
        node[below right,black]{\small{$\ball{x_0}{\rx}$}};        
\draw(9.8,2.0) 
        node[above right,black]{\small{$\ball{u_0}{\ru}$}};
\end{scope}
\end{tikzpicture}\\
\caption{\centering Constraint set $\const$ and the $n$-\emph{cube} around $x_0$ and $u_0$.} 
\label{constset}
\end{center}
\end{figure}
It is assumed that the constraints are not redundant, and $x$ and $u$ have the same dimension. We are interested in the following problem. 
\begin{problem}
\label{quadsolve}
For a given $u_0$, let there be an $x_0\in \const$ such that $(x_0,u_0)$ solves \eqref{quadeq}. Find the \emph{robustness margin}, i.e., the largest $\ru\correction{>} 0$ such  that for all $u$ satisfying $\Big\lvert[u_0]_i-[u]_i\Big\rvert\correction{<} \ru$ for all $i\in \n$, there exists an $x\in \const$ \correction{(not necessarily unique)} such that $(x,u)$ solves \eqref{quadeq}.
\end{problem}
First we utilize the ImFT to show that a nonzero robustness margin exists and then use the bounds derived for ImFT to come up with a lower bound on $\ru$. Casting \eqref{quadeq} as \eqref{non_eq}, define 
\begin{equation}
\label{f_quad_eq}
\R[n]\times\R[n]\ni    (x,u)\mapsto f(x,u) \coloneqq Q(x)+Lx-u\in\R[n].
\end{equation}
The Jacobian for $f$ at $(x_0,u_0)$ is given by
\begin{equation}
    \D f(x_0,u_0) = \pmat{\pdf{f}{x}(x_0,u_0),&\pdf{f}{u}(x_0,u_0)} = \pmat{2Q'(x_0)+L,&\eye{n}}
\end{equation}
where $Q'(x_0)$ denotes an ${n\times n}$ matrix with its $i^{\text{th}}$ row as $x_0^\top Q_i$ and $\eye{n}$ denotes an identity matrix of order $n$.  Since the inequalities are placed component-wise, i.e., we are looking for an $n$-\emph{cube} around $u_0$, and $\const$ \correction{is defined by \eqref{eq4.5}}.
We shall choose the infinity norm to work with, for a vector $v\in\R[n]$ the infinity norm is defined as
\begin{equation*}
    \norminf{v} \coloneqq \maxi{\lvert v_i\rvert}{i\in \n}. 
\end{equation*}
\correction{The corresponding induced norm is defined as follows}: let $A\:\R[n]\times\R[m]$ be a linear operator, then $\norminf{A}$ is defined as
\begin{equation*}
    \norminf{A} \coloneqq \maxi{\norminf{Ax}}{\norminf{x}\leq1}.
\end{equation*}
Let $\R[n]\times\R[n]\ni(u,v)\mapsto B(u,v)\in\R[m]$ \correction{be a bilinear map}. The induced norm on $B$ is defined as
\begin{equation*}
    \norminf{B} \coloneqq \maxi{\norminf{B(u,v)}}{\norminf{u}\leq1,~\norminf{v}\leq 1}.
\end{equation*}
With this premise set, we have the following result.
\begin{thm}
\label{th4.6}
Suppose $Q\in\R[n\times n] \And L\in\R[n\times n]$ is such that $2Q'(x_0)+L$ is non-singular, and $\const$ is nondegenerate, i.e., has a nonempty interior. Then \eqref{quadeq} is robustly solvable on $\const$. Moreover,
define $M_x\coloneqq \norminf{(2Q'(x_0)+L)^{-1}}$ and $K_{xx}\coloneqq \maxi{{\lvert x^\top Q_ix\rvert}}{\norminf{x}=1,i\in \n}$. Then for a given $0<\rx<\frac{1}{M_xK_{xx}}$ such that $\ball{x_0}{\rx}\subset \const$, the robustness margin $\ru$ is lower bounded by $\frac{\rx(2-M_xK_{xx}\rx)}{2M_x}$. 
\end{thm}

\noindent The proof is a straight forward application of \correction{Theorem} \ref{imft_est} on $f$ as defined in \eqref{f_quad_eq} and hence omitted.
\begin{rmk}
\correction{The condition $0<\rx<1/M_xK_{xx}$ ensures uniqueness of $x\in\ball{x_0}{1/M_xK_{xx}}$. However, one can not assert the uniqueness of $x$ on $\Omega$, using Theorem \ref{th4.6}.}   
\end{rmk}
\begin{examplenn}
In order to illustrate the above calculated bounds we present a simple example. The example is taken from \cite{dvijotham2019robust} so as to establish comparisons. The data is as follows: 
\begin{equation*}
    A = \bmat{-1&0\\\correction{1}&0\\0&\correction{-1}\\0&1},~B=\bmat{-0.5\\3\\-0.5\\3},~Q_1 =\bmat{1&0\\0&0}~,Q_2=\bmat{0&0\\0&1},~L=\bmat{1&-3\\2&-1},\text{ and}~u_0 = \bmat{-2\\4}.
\end{equation*}
The underlying expression is then given by
\begin{equation}
\label{q_exmpl}
    F(x) = \bmat{x_1^2+x_1-3x_2\\x_2^2+2x_1-x_2}=\bmat{-2\\4}\quad\text{with}\quad\const = \left\{x \:\bmat{0.5\\0.5}\leq\bmat{x_1\\x_2}\leq\bmat{3\\3}\right\}.
\end{equation}
The unique solution is given by $x_0 \approxeq \bmat{1.36&1.74}^\top$. One can check that $\D F(x_0)$ is nonsingular. The corresponding quantities are $L_x= 6.7204$, $M_x= 0.3763$, $L_u = 1$, and $K_{xx} = 2$. For $\correction{\rx} = 0.86$, $\ball{x_0}{0.86}\subset\const$. The robustness margin for \eqref{q_exmpl} comes out to be, \correction{$\ru \geq 1.546$}. \correction{This is greater than the estimates given by LP-feasibility  routine ($r_u\geq1.2054$) by Dvijotham et al. in  \cite{dvijotham2019robust}. The bounds can be further improved by changing the map $(x,u) \mapsto (F(x)-u)$ to $\bar {F}(x,u) \coloneqq (\D F(x_0)^{-1})\cdot(F(x)-u)$. This gives us an updated bound of $\ru\geq 1.5781$, however, this is smaller than the estimate given by the bound tightening procedure ($r_u\geq1.7069$) in \cite{dvijotham2019robust}. Figure \ref{kdresult} shows simulation results for the improved bounds. The blue scatter plot shows $F(\Omega)$. The yellow rectangle shows the $\bd\ball{u_0}{\ru}$. One can see that $\ball{u_0}{\ru}$ is contained in $F(\Omega)$ and thus for each $u\in\ball{u_0}{\ru}$ there will be an $x\in\omega$ such that $F(x)-u=0$. The bounds can be improved if one restricts the variation along a particular subspace. This is shown by the purple-colored parallelogram. The $u$ was successively allowed to vary along an arbitrary unit vector $(w_1,w_2)$, and the maximum perturbation was computed. This is denoted by $\bd(B_W)$ in the figure.} 
\begin{figure}
\begin{center}
\includegraphics[width=\linewidth]{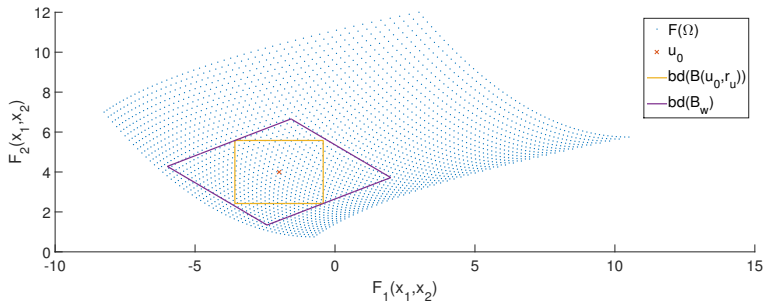}    
\caption{ \centering Robustness Margin for Example \ref{q_exmpl}}
\label{kdresult}
\end{center}
\end{figure}
\end{examplenn}
\subsubsection{Application to Power Systems}
Power Flow studies and Optimal Power Flow studies require us to calculate operating voltages and currents by balancing the load demand and the power generated. Integration of renewable sources like solar and wind energy in the power system network leads to uncertainty in the generated supply. Due to this uncertain supply of renewable sources, it is difficult to ensure that there will be sufficient power generation to meet the power demand while adhering to the stable operation limits of the power system network. The Power Flow equations involve balancing the generated power supply and the required demand. As described by Dvijotham et al. \cite{dvijotham2017solvability}, the AC Power Flow equations are written as follows
\begin{equation}
\label{ac_flow}
\begin{split}
\real\left(V_i\conjugate{Y_{i0}V_0}+\sum_{k=1}^nV_i\conjugate{Y_{ik}V_k}\right)&=p_i\quad\forall i\in\PQ,\\
\imaginary\left(V_i\conjugate{Y_{i0}V_0}+\sum_{k=1}^nV_i\conjugate{Y_{ik}V_k}\right)&=\correction{q_i},\quad\forall i\in\PQ,\\
\real\left(V_i\conjugate{Y_{i0}V_0}+\sum_{k=1}^nV_i\conjugate{Y_{ik}V_k}\right)&=p_i\quad\forall i\in\PV,\\
\lvert V_i\lvert^2&=v_i^2\quad\forall i\in\PV,
\end{split}
\end{equation}
where $V_i$ denotes the complex voltage phasor, $p_i,~q_i$ denotes the active and reactive power injection at the node $i$, and $Y$ denotes the admittance matrix. $\PV$ denotes the set of generator buses, and $\PQ$ denotes the load bus. $v_i$ denotes the root mean square (rms) voltage at the $i^\text{th}$ bus. $V_0$ denotes the reference slack bus and is kept fixed at $1+j0$ per unit magnitude, where $j=\sqrt{-1}$. Power Flow as given by \ref{ac_flow} is quadratic. Further, for stable operation of the power system network, one needs to maintain the bus voltage magnitude within prescribed limits. This imposes a quadratic constraint on \eqref{ac_flow}. 
Setting
\begin{align*}
    x &=\bmat{\real(V_1)&\ldots&\real(V_n)&\imaginary(V_1)&\ldots&\imaginary(V_n)}\text{ and}\\
    u &= \bmat{p_1&\ldots&p_n&q_1&\ldots&q_n&v_1^2&\ldots&v_n^2}
\end{align*} 
the powerflow equations given by \eqref{ac_flow} can be written as QCQP of type \eqref{quadsolve} with $\R[2n]\ni x\mapsto F(x)\in\R[2n]$ with its $i^\text{th}$ component given by 
\begin{equation}
\begin{split}
[F(x)]_i&\coloneqq \real\left(\sum_{k=1}^nV_i\conjugate{Y_{i0}V_0}+\conjugate{V_iY_{ik}V_k}\right)\quad\forall i\in\n,\\
[F(x)]_{n+i}&\coloneqq \imaginary\left(V_i\conjugate{Y_{i0}V_0}+\sum_{k=1}^nV_i\conjugate{Y_{ik}V_k}\right)\quad\forall i\in\PQ, \and\\
[F(x)]_{n+i}&\coloneqq{\real (V_i)}^2 +{\imaginary (V_i)}^2\quad\forall i\in \PV.
\end{split}
\end{equation}
For stable operation, the voltage magnitudes are to be kept within tolerance limits. For a fixed $\rx\in]0,1[$, the constraint set is 
\begin{equation*}
    \const = \{x\in\R[2n]\mid 1-\rx\leq v_i^2\leq 1+\rx ~ \text{for all}~i\in\n\}.
\end{equation*}
The constraint set $\const$ can be relaxed to the following
\begin{equation*}
    \const' = \{x\in\R[2n]\mid -\rx\leq \norminf{x-x^*}\leq \rx ~ \text{for all}~i\in\n\}
\end{equation*}
where $x^*\in\R[2n]$ be the nominal operating point satisfying $(\real{V_i^*})^2+\imaginary({V_i^*})^2=1$ for all $i\in\n$. \\ \\
The test cases are obtained from the dataset given in the MatPower package found in MATLAB software \cite{zimmerman2010matpower}. The package contains well-defined libraries and a set of routines for solving problems like Power Flow analysis, Optimal Load Flow, and DC Power Flow analysis.  The package is open source and is available online. The problem is simulated for several test cases from MatPower. The maximum allowable $\rx$ and $\ru$ are recorded in P.U. magnitude. In order to simulate real-life scenarios, we only consider the variation in the first five dimensions of $u$. The results are tabulated in Table \ref{tab:results}, where $M_f\coloneqq\norm{\D F^{-1}(u_0)}$, $M_F' \coloneqq\supr{\D F^{-1}(u_0)\cdot \Tilde{u}}{u\in\R[2n],~\tilde u_i=0~\text{for all}~i>5}$, and $K_{\bar F} = \supr{\norm{\D^2\bar{F}(x)}}{x\in\R[2n]}$ where $\bar F = (\D F(x_0))^{-1}\cdot F$. \correction{For Case 5, 9, 14, and 30, we also plot $\rx$ with respect to $\ru$ and compare it with the bounds given by the bound tightening method given by Dvijotham et al. \cite{dvijotham2019robust}. These are shown in Figure. \ref{powerresult} } 
\begin{table}[h]
\centering
\begin{tabular}{@{}cccccc@{}}
\toprule
Case &$M_F$ &$M_F'$ &$K_{\bar F}$ & Max $\rx$ & Max $\ru$ \\ 
\midrule
5    &0.5154 &0.0512  &12.971           &0.0771             &0.7529\\
9    &1.3802 &0.7968  &39.065           &0.0256             &0.0161\\
14   &2.4795 &0.5291  &91.066           &0.0110             &0.0104\\
30   &6.2576 &0.3225  &0.330$\times10^3$       &3.303$\times 10^{-3}$      &5.120$\times 10^{-3}$\\
57   &12.153 &0.2657  &1.032$\times10^3$       &0.969$\times 10^{-3}$      &1.823$\times 10^{-3}$\\
85   &5.1119 &0.0140  &13.48$\times10^3$       &0.074$\times 10^{-3}$      &2.643$\times 10^{-3}$\\
141  &22.049 &0.2371  &1.958$\times10^3$       &0.512$\times 10^{-3}$      &1.078$\times 10^{-3}$\\
\bottomrule
\end{tabular}
\caption{\centering Lower bounds on the robustness margin $\ru$ for various cases from MatPower}
\label{tab:results}
\end{table}
\begin{figure}
\begin{center}
\includegraphics[width=\linewidth]{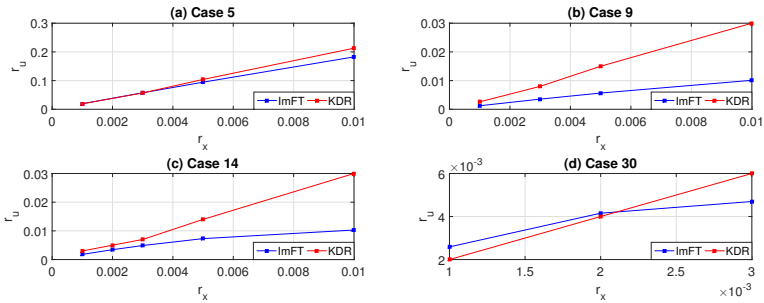}    
\caption{\centering Robustness margin estimates obtained using ImFT (Blue) for (a) Case 5, (b) Case 9, (c) Case 14, and (d) Case (30) compared with estimates obtained using bound tightening (Red) from \cite{dvijotham2019robust}.}
\label{powerresult}
\end{center}
\end{figure}
\begin{rmk}
\correction{From Table \ref{tab:results} and Figure \ref{powerresult}, one can see, for lower values of $\rx$, our bounds are comparable to those of \cite{dvijotham2019robust}. However, with an increase in the dimension, our estimates start to perform poorly as compared to those given in \cite{dvijotham2019robust}. The one key reason for this is the way we bound $K_{xx}$. In order to keep computation minimum we bound $\maxi{\lvert x^\top Qx\rvert}{\norminf{x}\leq1}$ by $\sum_{i,n\in\n}\lvert Q_{ij}\rvert$.  For a general matrix $Q$ this grows quadratically with respect to $n$. However, if $Q$ is a sparse matrix, then one can obtain tighter estimates. One can alternatively utilize the inequlaity $$\maxi{\lvert x^\top Qx\rvert}{\norminf{x}\leq1}\leq j\norm{Q}_2,$$to estimate  which grows linearly with $n$.}

Further, with an increase in $n$, the problem dimension increases and norm-based inequalities such as triangular inequality and sub-multiplicative inequality become more conservative.  Our estimates are conservative when compared to those provided in \cite{dvijotham2019robust}. However, the numerical computation required to compute these bounds is much lower than that required by \cite{dvijotham2019robust}. In particular, no optimization routine is required to compute these bounds. \correction{The bounds given by Dvijotham et.al. in \cite{dvijotham2019robust} require us to perform nontrivial optimizations on matrices and solve LP feasibility programs. In comparison to that, we only require to compute for the inverse of an $n\times n$ matrix i.e., $\D F(x_0)$, for which there exist efficient algorithms and packages in literature. Since our objective is not to tackle matrix inversion methods, we do not investigate methods for computing matrix inversion. Some literature on matrix inversion can be found in \cite{osti_5257971, MARTINSSON2005741, lisong,liujun}}.
One key limitation of this method is that it can only provide us with a lower bound on the robustness margin.
\end{rmk}
\subsubsection{Robustness of the Algebraic Riccati Equation}
Consider a linear time invariant control system
\begin{equation}
\label{lin_are}
\odv{x}{t} =  Ax(t)+Bu(t)
\end{equation}
with $x(t)\in\R[n]$, $u(t)\in\R[m]$ for all $t>0$. For any given $n\in\N$, $\symm{n}$, $\poss{n}$, and $\psd{n}$ denotes the set of symmetric matrices, positive definite matrices, and positive semidefinite matrices of order $n\times n$ respectively.  Let $$\mathcal{U} = \bigg\{[0,\infty[\ni t\mapsto u(t)\in\R[m]\mid\int_{0}^\infty\norm{u(t)}^2\d t<\infty\bigg\}~,$$ be the set of square integrable controls over \correction{infinite} horizon. Define the infinite horizon Linear Quadratic Regulator (LQR) as the following optimal control problem\correction{.}
\begin{problem}
For given $Q\in\psd{n}, \correction{\And} R\in\poss{m}$ compute a \correction{$u^*$,} if it exists, such that it solves
\begin{equation*}
\begin{cases}
\underset{u\in\mathcal{U}}{\minimize} \quad    &J(u)=\frac{1}{2}\int_{0}^{\infty}x(t)^\top Qx(t)+u(t)^\top Ru(t),\\
\\
\text{subject to}\quad & x(0) = x_0, \text{ and } \eqref{lin_are} \text{ for almost all }t\in[0,\infty[~.
\end{cases}
\end{equation*}
\end{problem}
\noindent From results on LQR \cite[Chapter~6]{liberzon2011calculus}, the optimal control is of \correction{the} type 
\begin{equation*}
    u^*(t) = -R^{-1}B^\top Px^*(t) ~ \quad \text{ for almost all  }t\in[0,\infty[ \correction{,}
\end{equation*}
where $P\in\poss{n}$ solves the Algebraic Riccati Equation (ARE)
\begin{equation}
\label{riccatieq}
A^\top P +PA + Q - PBR^{-1}B^\top P = 0\correction{.}   
\tag{ARE}
\end{equation}
 For given $A_0\in\R[n\times n]$, \correction{and} $B_0\in\R[n\times m]$ such that $(A_0,B_0)$ forms a stabilizable pair, \correction{consider a linear control system} given by 
 \begin{equation}
\label{perturbedsys}
    \odv{x}{t} = (A_0+\dstb)x + B_0u
\end{equation}
where parameter $\dstb\in\R[n\times n]$ denotes the unmodelled part of system dynamics. 
\begin{problem}
\label{ret_sys_rob}
Find a bound on $\norm{\dstb}$ such that \eqref{perturbedsys} is stabilizable, i.e., find an $\rl$ such that $(A_0+\dstb,B_0)$ is stabilizable for all $\dstb\in\ball{0}{\rl}$.
\end{problem}
\begin{rmk}
Note that we may have very well considered variations in $B_0$ as well, however for illustration purposes, we restrict to variations in $A_0$.
\end{rmk}
\begin{lem}
\label{lem_are}\emph{\cite[Chapter~6]{liberzon2011calculus}}
For each $Q\in\poss{n}$ and $R\in\poss{m},$ a unique positive definite solution for \eqref{riccatieq} exists if and only if $(A,B)$ forms a stabilizable pair\footnote{\correction{Two given matrices $A\in\R[n\times n]$ and $B\in\R[m\times n]$ are said to form a stabilizable pair if there exists a $K\in\R[n\times m]$ such that $A+BK$ is Hurwitz i.e., have eigenvalues with strictly negative real parts.}}.
\end{lem}
Using the above lemma, for any given $\dstb$, $(A_0+\dstb,B_0)$ being stabilizable is equivalent to \correction{showing} that \eqref{riccatieq} has a positive definite solution with $A= A_0+\dstb$. Further, since $\poss{n}$ is open in $\symm{n}$, there exists an $\rp$ such that for all $P\in\ball{P_0}{\rp}\subset\symm{n}$, $P$ is positive definite. Problem \ref{ret_sys_rob} can now be formulated as the following robustness problem\correction{.}
\begin{problem}
\label{ricc_prob_rob}
For given $A_0,B_0,Q,\and R$, let there be a $P_0\in\poss{n}$ such that it solves \eqref{riccatieq}. For a given $\rp>0$ such that \correction{$\ball{P_0}{\rp}\subset\poss{n}$}, find an $\rl>0$ such that for each $A\in\ball{A_0}{\rl}\subset\R[n\times n]$ there \correction{exists} a $P\in\ball{P_0}{\rp}$ solving 
\begin{equation*}
f(A,P) \coloneqq A^\top P +PA + Q - PB_0 R^{-1}B_0^\top P = 0\correction{.}
\end{equation*}
\end{problem}
\noindent Note that Problem \ref{ricc_prob_rob} is just investigating the robustness of solutions of nonlinear equations. \correction{Therefore we may use the bounds derived for ImFT to compute $\rl$}. To this note we have $\pdf{f}{P}(A,P))$ defined as 
\begin{equation}
\symm{n}\ni \mu\mapsto\pdf{f}{P}(A,P)\cdot \mu = A^\top \mu+\mu A-(\mu B_0 R^{-1} B_0^\top P +PB_0 R^{-1}B_0^\top \mu)   
\end{equation}
and $\pdf{f}{A}(A,P)$ is defined as
\begin{equation}
\R[n\times n]\ni \nu\mapsto \pdf{f}{A}(A,P)\cdot\nu = \nu^\top P+P\nu.
\end{equation}
Evaluating $\pdf{f}{P}$ at $(A_0,P_0)$ we have the following linear map \correction{:}
\begin{equation}
\begin{split}
\mu\mapsto\pdf{f}{P}(A_0,P_0)\cdot \mu &=  A_0^\top \mu +\mu A_0 - (\mu B_0 R^{-1}B_0^\top P_0 +P_0 B_0 R^{-1}B_0^\top \mu)\\
&= (A_0-B_0R^{-1}B_0^\top P_0)^\top \mu + \mu(A_0-B_0R^{-1}B_0^\top P_0)\\
&= A_{c}^\top \mu+\mu A_{c}
\end{split}
\end{equation}
with $A_{c} = A_0-B_0R^{-1}B_0^\top P_0$ Since $(A_0,B_0)$ is a stabilizable pair and $u^* = -R^{-1}B_0^\top P_0x^*$ is the stabilizing control, $A_{c}$ is a stable matrix with its \correction{eigenvalues} having strictly negative real parts. For any $v\in\symm{n}$ define 
\begin{equation*}
    v\mapsto \mu(v)\coloneqq -\int_{0}^{\infty}e^{tA_{c}^\top}ve^{tA_{c}}\d t\correction{,}
\end{equation*}
then we have 
\begin{equation*}
    A^\top_{c}\mu(v)+\mu(v)A_{c} = v\correction{.}
\end{equation*}
Therefore, $\pdf{f}{P}(A_0,P_0) : \symm{n}\lra \symm{n}$ is a surjective map. \correction{Furthermore, using} rank-nullity theorem we can conclude that it is an isomorphism with inverse
\begin{equation}
    v\mapsto (\pdf{f}{P}(A_0,P_0)^{-1})\cdot v = -\int_{0}^{\infty}e^{tA_{c}^\top}ve^{tA_{c}}\d t\correction{.}
\end{equation}
Since $\pdf{f}{P}(A_0,P_0)$ is an isomorphism, we can now apply Lemma \ref{imftc1}. Define $$M_{P}\coloneqq \norm{\pdf{f}{P}(A_0,P_0)} \qand L_{A} \coloneqq \norm{\pdf{f}{A}(A_0,P_0)}.$$ Further for any given $r>0$ and $r^*>0$\correction{,} define
\begin{equation*}
    \lipconst[P]{r,r^*} \coloneqq \supr{\pdf{f}{P}(A,P)-\pdf{f}{P}(A_0,P_0)}{(A,P)\in\ball{\correction{A_0}}{r}\times\ball{P_0}{r^*}}
\end{equation*}
and
\begin{equation*}
    \correction{\lipconst[A]{r}\coloneqq \supr{\pdf{f}{A}(A,P_0)-\pdf{f}{A}(A_0,P_0)}{A\in\ball{A_0}{r}}.}
\end{equation*}
Then for all $\rp>0,\and \rl>0$ with $\ball{P_0}{\rp}\subset\poss{n}$ satisfying
\begin{subequations}
\label{are_eps_c1}
\begin{align}
\lipconst[A]{\rl}\rl+\lipconst[P]{\rl,\rp}\rp&<\frac{\rp}{M_P}-L_A \rl,\correction{\And}\\
M_P\lipconst[P]{\rl,\rp}&< 1\correction{,}
\end{align}      
\end{subequations}
for each $A\in\ball{A_0}{\rl}$\correction{,} there exists a $P\in\ball{P_0}{\rp}$ such that it solves \ref{riccatieq} and consequently $(A,B)$ is stablizable for all $A\in\ball{A_0}{\rl}$.
\begin{rmk}
The bounds $\rp$ and $\rl$ depend upon the choice of the nominal point $(A_0,P_0)$ which itself depends upon the
choice of $Q$ and $R$. One can then optimize these bounds over all possible choices of $Q\in\poss{n}$ and $R\in\poss{m}$. However, for each $Q$ and $R$, to compute the corresponding $P$, we need to solve \eqref{riccatieq} which is a nonlinear equation making the optimization of $\rp,\rl$ with respect to $Q$ and $R$ at least difficult, if not intractable.
\end{rmk}
\noindent We now present a simple illustrative example on \correction{the} double integrator.
\begin{examplenn}
Consider the (perturbed) double integrator system defined by 
\begin{equation*}
    \odv{}{t}\pmat{x_1\\x_2}= \left(\pmat{0&1\\0&0}+\dstb\right)\pmat{x_1\\x_2} + \pmat{0\\1}u
\end{equation*}
with $\dstb\in\R[2\times 2]$. Set $Q = \eye{2}$ and $R=1$, then the corresponding solution of \eqref{riccatieq} is $$P_0 = \pmat{\sqrt{3}&1\\1&\sqrt{3}}.$$ For any $\pmat{v_{11}& v_{12}\\v_{21}&v_{22}}\eqqcolon v\in\R[2\times 2]$ and $\pmat{\mu_1&\mu_2\\\mu_2&\mu_3}=:\mu\in\symm{2}$ we have
\begin{equation*}
    \pdf{f}{A}(A_0,P_0)\cdot v = \pmat{2(\sqrt{3}v_{11}+v_{21})&v_{11}+v_{22}+\sqrt{3}(v_{21}+v_{12})\\v_{11}+v_{22}+\sqrt{3}(v_{21}+v_{12})&2(\sqrt{3}v_{22}+v_{12})}
\end{equation*}
and
\begin{equation*}
    \pdf{f}{P}(A_0,P_0)\cdot \mu = \pmat{-2\mu_2&\mu_1-\sqrt{3}\mu_2-\mu_3\\\mu_1-\sqrt{3}\mu_2-\mu_3&2(\mu_2-\sqrt{3}\mu_3)}.
\end{equation*}
Since $A_c = (A_0-B_0B_0^\top P_0)$ is \correction{Hurwitz}, $\pdf{f}{P}(0,P_0)$ is invertible and for any $\pmat{\nu_1&\nu_2\\\nu_2&\nu_3}=:\nu\in\symm{2}$ we have
\begin{equation}
    \pdf{f}{P}(A_0,P_0)^{-1}\cdot\nu = \pmat{-1.1574\nu_1+\nu_2-0.2887\nu_3&-0.5\nu_1\\-0.5\nu_1&-0.2887(\nu_1+\nu_3)}. 
\end{equation}
Using the infinity norm, we have
\begin{equation*}
    L_A= \norm{\pdf{f}{A}(A_0,P_0)} = 6.928   \qand M_P = \norm{\pdf{f}{A}(A_0,P_0)^{-1}} = 3.0207
\end{equation*}
Further, we have
\begin{equation*}
\lipconst[A]{\rl}=0 \qand \lipconst[P]{\rl,\rp}\leq 2(\rp+\rl). 
\end{equation*}
Substituting in Equation \eqref{are_eps_c1}, \correction{we get that} for all $\rp>0,\correction{\And}\rl>0$ satisfying 
\begin{align*}
    2\rp(\rl+\rp)&< 0.3310\rp-6.928\rl \qand\\
    \rp+\rl&<0.1655 
\end{align*}
for each $A\in\ball{A_0}{\rl}$\correction{,} there exists a $P\in\ball{P_0}{\rp}$ such that it solves the \eqref{riccatieq}.
\end{examplenn}
\begin{rmk}
The purpose of these examples is not to compute the optimum bounds for which we require nonlinear optimization (often global), but to rather show that several fundamental problems in systems and control applications can be recast into the ImFT and IFT framework. 
\end{rmk}
\subsection{Estimating the Domain of Feedback linearizability}
Let $\stateset\subset\R[n]$ and $\controlset\subset\R[m]$ be nonempty open sets and $\stateset\times\controlset\ni(x,u) \mapsto f(x,u)\in\stateset$ be an analytic map. Define a discrete time control system as
\begin{equation}
\label{cont_system}
x(k+1) = f(x(k),u(k))
\end{equation}
where for each $k\in \N$, $x(k)\in\stateset$ denotes the system state and $u(k)\in\controlset$ is the control input. A point $(\xopt,\uopt)\in\stateset\times\controlset$ is called the \emph{equilibrium point} of \eqref{cont_system} if it satisfies $f(\xopt,\uopt) = \xopt$. 
\begin{defn}
Let $(\xopt,\uopt)$ be an equilibrium point of \eqref{cont_system} and $\Open{\xopt}\ni\xopt$, $\Open{\uopt}\ni\uopt$ \correction{be nonempty and open. Suppose} 
\begin{equation*}
    \Open{\xopt}\ni x\mapsto \fstate(x) \in \R[n]
\end{equation*}
is a diffeomorphism onto its image and 
\begin{equation*}
    \Open{\xopt}\times\Open{\uopt}\ni (x,u) \mapsto \fcont(x,u)\eqqcolon v\in\R[m]
\end{equation*}
is such that for each $x\in\Open{\xopt}$, $\fcont(x,\cdot):\Open{\uopt}\lra\R[m]$ is an injective map. Then \eqref{cont_system} is said to be \emph{locally feedback linearizable} on $\Open{\xopt}\times\Open{\uopt}$ if $\fstate$ and $\fcont$ transforms \eqref{cont_system} to an equivalent controllable linear system of type
\begin{equation}
\label{system_lin}
z(k+1)\coloneqq Az(k)+Bv(k)
\end{equation}
with $z(k)=\fstate(x(k))$ and $v(k) = \fcont(x(k),u(k))$ for all $k\in\N$.
\end{defn} 
Since feedback-linearization-based methods are local, i.e., the results \correction{hold} only in a small neighborhood of \correction{the} operating point, for practical implementation of such methods one needs to know an apriori estimate on the domain on which \eqref{cont_system}  is feedback linearizable. \correction{Using the necessary and sufficient conditions for feedback linearizability one can ensure the existence of such open sets and maps. However, these methods do not provide estimates about how large these sets are.} For systems with $f$ being analytic, this is equivalent to \correction{finding} the domains on which $\fstate$ is \correction{a} diffeomorphism and $\fcont(x,\cdot)$ is injective. One can therefore utilize the bounds on the ImFT and IFT to come up with these estimates. The following propositions provide first estimates on the domain on which $\fstate$ \correction{is a diffeomorphism} and $\fcont(x,\cdot)$ \correction{is injective.} 
\begin{prop}
\label{statebound}
Let $\fstate$ be the linearizing coordinate change for \eqref{cont_system}. Define $$L_\fstate \coloneqq \norm{\D \fstate(\xopt)}\and M_\fstate \coloneqq\norm{\D\fstate^{-1}(\zopt)}.$$ For any given $r>0$ define
$$ \Lipconst[\fstate]{r} \coloneqq \supr{\D\fstate(x)-\D\fstate(x_0)}{x\in\ball{x_0}{r}}$$
and define 
\begin{equation*}
    P_\fstate \coloneqq \arg\maxi{\epsilon(r)}{r>0, \Lipconst[\fstate]{r}\leq\frac{1}{M}}~\correction{\And}~ P_\fstate'= \epsilon(P_\fstate)\correction{,}
\end{equation*}
where for all $r>0$, $\epsilon(r) \coloneqq r(1-M_\fstate\Lipconst[\fstate]{r})/M_\fstate.$ \correction{
Then there exists an open set $H_\fstate\subset\ball{\xopt}{P_\fstate}$ such that $\fstate$ maps $H_{\fstate}$ onto $\ball{\xopt}{P_\fstate'}$ diffeomorphically.}
\end{prop}
\begin{prop}
\label{controlbound}
Let $u\mapsto\fcont(x,u)\coloneqq v$ be the new linearized control. Suppose $\pdf{\fcont}{u}(\xopt,\uopt)$ is nonsingular, then define 
\begin{equation*}
  M_\fcont^u \coloneqq \norm{\pdf{\fcont}{u}(\xopt,\uopt)^{-1}}\and L_\fcont^x \coloneqq\norm{\pdf{\fcont}{x}(\xopt,\uopt)}.  
\end{equation*}
For a given $\rx>0$, $\ru>0$ and $\rv>0$ let $\mathcal{B}_1 \coloneqq \ball{\xopt}{\rx}\times\ball{\uopt}{\ru}$ and define 
\begin{equation*}
\correction{\lipconstr[\fcont]{\rx}{x} \coloneqq \supr{\norm{\pdf{\fcont}{x}(\xopt,\uopt)-\pdf{\fcont}{x}(x,\uopt)}}{(x,\uopt)\in\mathcal{B}_1}} 
\end{equation*}
and 
\begin{equation*}
\lipconstr[\fcont]{\rx,\ru}{u}\coloneqq \supr{\norm{\pdf{\fcont}{u}(\xopt,\uopt)-\pdf{\fcont}{u}(x,u)}}{(x,u)\in\mathcal{B}_1}.
\end{equation*}
Then for any given $\rx>0$, $\ru>0$  and $\rv>0$ satisfying \begin{subequations}\label{eqn:epsbnd1}
\begin{align}
\correction{\lipconstr[\fcont]{\rx}{x}\rx}+ \lipconstr[\fcont]{\rx,\ru}{u}\ru&<\frac{\ru}{M_\fcont^u}-\rx L_\fcont^x-\rv,\correction{\And}\\
M_\fcont^u\lipconstr[\fcont]{\rx,\ru}{u}&<1
\end{align}
\end{subequations}
for each $(x,v)\in\ball{\xopt}{\rx}\times\ball{\vopt}{\rv}$ there exists a unique $u\in\ball{\uopt}{\ru}$ satisfying $$\fcont(x,u) = v.$$ Further, define $\ball{\xopt}{\rx}\times\ball{\vopt}{\rv}\ni(x,v)\mapsto\gamma(x,v) \correction{\coloneqq} u\in\ball{\uopt}{\ru}$, then $\gamma$ is an analytic map.
\end{prop}
\correction{Propositions} \ref{statebound} and \ref{controlbound} are direct applications of bounds computed on the domain of IFT and ImFT. \correction{Note that, Propositions} \ref{statebound} and \ref{controlbound} provide the domain on which $\fstate$ and $\fcont$ are well defined. However, \correction{for the control system} \eqref{cont_system} to be feedback linearizable one also needs to ensure that the trajectory \correction{belongs} to $\Open{\xopt}$ (or equivalently in $\Open{\zopt})$ for all $k\in\N$. This limits the choice of \correction{the} control input $v(k)$. From the linearized dynamics, we have
\begin{equation*}
    z(k+1)-\zopt = A(z(k)-\zopt)+B(v(k)-\vopt)\correction{.}
\end{equation*}
Taking norms, we have
\begin{equation*}
\begin{split}
\norm{z(k+1)-\zopt}&=\norm{A(z(k)-\zopt)+B(v(k)-\vopt)}\\
&\leq \norm{A}\norm{z(k)-\zopt}+\norm{B}\norm{v(k)-\vopt}\correction{.}
\end{split}    
\end{equation*}
For a given $0<\rz<P_\fstate'$, for all $0<\epsilon_v<\frac{1}{\norm{B}}\left(P_\fstate'-\norm{A}\rz\right)$ we have $\norm{z(k+1)-\zopt} < \rz.$ Thus, for all $z(k)\in\ball{\zopt}{\rz}$ and $v(k)\in\ball{\vopt}{\epsilon_v}$, $z(k+1)\in\ball{\zopt}{\rz}.$ Combining bounds from Proposition \eqref{statebound} and \eqref{controlbound}, along with ensuring that trajectories stay in $\Open{\zopt}$ we can find an $\Open{\xopt}$ and $\Open{\uopt}$ such that System \eqref{cont_system} is feedback \correction{linearizable} on $\Open{\xopt}\times\Open{\uopt}$. 
\begin{rmk}
The bounds are given by \correction{Propositions} \ref{statebound} and \ref{controlbound} \correction{utilize} minimum information of $\fstate$ and $\fcont$\correction{. In particular, they only rely} on the bounds on the first-order derivatives. This leads to often conservative estimates. However, the bounds can be significantly improved if we utilize the structure of $\fstate$ and $\fcont$ as we demonstrate in the next example. 
\end{rmk}
\begin{examplenn}
Consider the following discrete system
\begin{equation}
\label{discexample}
\pmat{x_1(k+1)\\x_2(k+1)} = \pmat{x_2(k)\\(1+x_1(k))^2u(k)}
\end{equation}
with an equilibrium point at $(0,0,0)$. Setting $\fstate(x_1,x_2) = (x_1,x_2)$ and $\fcont(x_1,x_2,u) = (1+x_1^2)u = v$, one can linearize \eqref{discexample} about $(0,0,0)$. In order to compute the bounds on the domain of feedback linearizability of \eqref{discexample} we employ Propositions \ref{statebound} and \ref{controlbound}. Let us first compute bounds for $\fstate$. \correction{Since $\fstate$ is the identity map, it is a global diffeomorphism on $\R[2]$.} A similar assertion is obtained from Proposition \ref{statebound}. From definition of $\fcont$ we have
\begin{equation}
    \pdf{\fcont}{u}(x_1,x_2,u) = (1+x_1)^2
\end{equation}
and 
\begin{equation}
    \pdf{\fcont}{x}(x_1,x_2,u) = \pmat{2(1+x_1)u&0}
\end{equation}
Computing quantities as defined in Proposition \ref{controlbound} we have $M_\fstate^{u} \coloneqq \norm{\pdf{\fcont}{u}(0,0,0)} = 1$, $L_\fstate^{x} \coloneqq \norm{\pdf{\fcont}{x}(0,0,0)} = 0$. For any given $\rx>0,\ru>0$, $\mathcal{B}_1 \eqqcolon\ball{x_0}{\rx}\times\ball{x_0}{\ru}$, we have
\begin{equation*}
\lipconstr[\fcont]{\rx}{x} \coloneqq \supr{\norm{\pdf{\fcont}{x}(\xopt,\uopt)-\pdf{\fcont}{x}(x,\uopt)}}{(x,\uopt)\in\mathcal{B}_1}  = 0
\end{equation*}
and
\begin{equation*}
\lipconstr[\fcont]{\rx,\ru}{u} \coloneqq \supr{\norm{\pdf{\fcont}{x}(\xopt,\uopt)-\pdf{\fcont}{x}(x,u)}}{(x,u)\in\mathcal{B}_1}  = \rx(\rx+2)
\end{equation*}
Substituting these \correction{quantities in to Equation \eqref{eqn:epsbnd} we get that,} for $\rx>0,\ru>0,\rv>0$ satisfying \begin{subequations}\label{eqn:flbound}
\begin{align}
\rv<\ru(1-\rx(\rx+2)),&\\
(2+\rx)\rx<1,&
\end{align}
\end{subequations}
for each $(x,v)\in\ball{x_0}{\rx}\times\ball{u_0}{\rv}$\correction{,} there exists a unique $u\in\ball{u_0}{\ru}$ such that $\fcont(x,u) = v$. Therefore \eqref{discexample} is feedback linearizable on $\ball{x_0}{\rx}\times\ball{u_0}{\ru}$. For $\rx = 0.2$, we have $\rv\leq 0.56\ru$. \end{examplenn}
However, the bounds given by \correction{Equation} \eqref{eqn:flbound} are conservative in nature. From direct observation, for $x_1\neq -1$, one has $u = v/(1+x_1)^2$ and hence the map $\fcont(x,\cdot)$ is globally invertible, further $\R\setminus\{-1\}\times\R\times\R \ni(x_1,x_2,v)\mapsto\gamma(x,v) =v/(1+x_1)^2$ is a smooth map and thus \eqref{discexample} is feedback \correction{linearizable} for all $(x_1,x_2,u)\in\R\setminus\{-1\}\times\R\times\R$\correction{.}
\section{Conclusion}
In this article, we provide a lower bound on the domain of the applicability of the ImFT. We utilize degree theoretic ideas to come up with these estimates. One key advantage that is gained by utilizing the degree theoretic ideas is the applicability of these results. For $\cont{2}$ functions the bounds are given as \correction{functions} of the magnitude of the first-order derivatives evaluated over a point and the bounds on the second-order derivatives calculated over a region of interest. These ideas can be suitably extended to $\cont{1}$ and $\cont{0}$ maps. \correction{These bounds can be suitably extended to IFT and the bounds derived by Corollary \ref{ift_est} surpass those given by Abraham et al.~\cite{abraham2012manifolds}.}

The ImFT and IFT have several applications in system theory and control. In this article, we have addressed a few of them. We utilize these bounds to investigate the robustness of the solutions of the nonlinear equations with parametric variations. The method is adapted to address Quadratically Constrained Quadratic \correction{Programs} (QCQPs). In control theory, the Algebraic Riccati Equation (ARE) can be formulated as a \correction{QCQP}. We utilize the bounds on ImFT to compute the robustness of the solutions of the ARE under variations in the system matrix. This helps us compute bounds on system matrices so that the system remains stable. We also apply these bounds on the \correction{power flow} equations to compute margins on the allowable power variations that ensure the stable operation of the power system network. We validate our results on the benchmark systems provided in the MatPower package of MATLAB software. 

Another important application of these bounds is in the feedback linearization methods, where we use these bounds to come up with an estimate on the \correction{domains} of the feedback linearizability of discrete-time systems. 

Note that we do not claim presented bounds to be optimal, rather we tend to seek estimates that require minimal numerical computation and can be used in scenarios where limited computation capacity is available on board. As a future extension, one may look at improving these estimates by exploiting the structures of the underlying mapping $f$ and finding different applications of these bounds in engineering problems.
\section*{Acknowledgments}
We would like to thank Professor Mayank Baranwal for his input and suggestions.
\section*{Declarations}
\subsection*{Ethical approval}
Not applicable.
\subsection*{Competing interests}
The authors declare that they have no competing interests whatsoever.
\subsection*{Authors' contribution}
\textbf{Ashutosh Jindal}: Formal analysis, Investigation, and Writing original draft; \textbf{Debasish Chatterjee}: Conceptualization, Methodology, and Supervision; \textbf{Ravi Banavar}: Supervision and Critical Analysis. 
\subsection*{Funding}
The authors have no funding information to declare.
\subsection*{Availability of data and materials}
Not applicable.
\begin{appendices}
\section{Some more results on the ImFT using Degree Theory}
\subsection{Showing \cite[Proposition 2.5.4]{abraham2012manifolds} as a corollary of Proposition \ref{Imft_holtzman}.}
\label{appen}
For $\cont{2}$ maps,  Abraham et al.~\cite{abraham2012manifolds} provide explicit estimates on the size of the neighborhoods involved in the IFT. These bounds are given in the following proposition. 
\begin{prop}\emph{\cite[Proposition~2.5.6]{abraham2012manifolds}} 
\label{propamd}
Suppose $f\: U\subset E\lra F$ is $\cont{r}$, $r\geq2$, $x_0\in U$, and $\D f(x_0)$ is an isomorphism. Let $L= \norm{\D f(x_0)}$ and $M=\norm{(\D f(x_0))^{-1}}$. Assume $\norm{\D^2f(x)}\leq K$ for all $x\in\ball{x_0}{R}\subset U$, for some $R>0$. Define $P\coloneqq \min\{\frac{1}{2KM},R\}$ and $P' = \frac{P}{2M}$. Further, let $N= 8M^3K$ and $Q \coloneqq \min\{\frac{1}{2NL},\frac{P}{2M},P\}$ and $Q'\coloneqq \frac{Q}{2L}$. Then there exist 
\begin{enumerate}[label=\textup{(\ref{propamd}\alph*)},leftmargin=*, widest=b, align=left]
    \item an open set $H\subset\ball{x_0}{P}$ such that $f$ maps $H$ diffeomorphically onto $\ball{x_0}{P'}$ and
\label{prpamd_1}    
    \item an open set $H'\subset\ball{f(x_0)}{Q}$ such that $f^{-1}$ maps $H'$ diffeomorphically onto $\ball{x_0}{Q'}$.
\label{prpamd_2}
\end{enumerate}
\end{prop}
Proposition \ref{propamd} can be shown as a corollary of Proposition \ref{Imft_holtzman} as we now demonstrate.
\begin{proof}[Proof of Proposition~{\upshape\ref{propamd}}] Define $(x,y)\mapsto F(x,y) = f(x)-y$. Then we have $F(x,y) =0 \iff f(x) = y$. Furthermore, $\pdf{F}{x}(x_0,f(x_0)) = \D f(x_0)$ is nonsingular. Therefore $F$ satisfies assumptions of Proposition \ref{Imft_holtzman}. Define $M = \norm{\D f(x_0)^{-1}}$ and let $R>0$ be such that $\ball{x_0}{R}\subset U$ and define
$$K = \supr{\norm{\D^2f(x)}}{x\in\ball{x_0}{R}}.$$
For a given $0<\rx<R$ and $\ry$, define 
\begin{equation*}
g_1(\rx) = K\rx \and g_2(\ry) = \ry.   
\end{equation*}
Then for all $(x,y)\in\cl\ball{x_0}{\rx}\times\ball{y_0}{\ry}$, we have
\begin{align*}
    \norm{\pdf{F}{x}(x,y)-\pdf{F}{x}(x_0,y_0)} & =  \norm{\D f(x)-\D f(x_0)}\leq g_1(\rx),\And\\
    \norm{F(x_0,y)} &= \norm{y_0-y}\leq \ry. 
\end{align*}
Setting $\alpha = 0.5$, for all $\rx,\ry$ satisfying 
\begin{equation*}
    MK\rx\leq0.5 \And M\ry\leq 0.5\rx
\end{equation*}
For each $y\in\ball{y_0}{\ry}$ there exists a unique $x_y$ in $\cl\ball{x_0}{\rx}$ satisfying $f(x_y) = y$. Defining $y\mapsto g(y) \coloneqq x_y$, provides the local inverse, further $g$ is $\cont{2}$ in view of Theorem \ref{ift}. Setting $\rx = P$, $\ry =P'$, and $H = g(\ball{y_0}{P'})$ proves \ref{prpamd_1}. From the expression of $\D f^{-1}(y)$, we have 
\begin{equation*}
    \supr{\D^2f^{-1}(y)}{y\in\ball{y_0}{P'}} \leq 8M^3K=N.
\end{equation*}
Therefore, replacing $f$ with $g$, $M$ with $L$ and $K$ with $N$ we get $Q = \min\{\frac{1}{2NL},\frac{P}{2M}, P\}$ and $Q'=\frac{Q}{2L}$, such that for each $x\in\ball{x_0}{Q'}$ there exists a $y_x\in\cl\ball{y_0}{Q}$ such that $f(x) = y_x$. Defining $H' = f(\ball{x_0}{Q'})$ proves \ref{prpamd_2}. 
\end{proof}
\subsection{Generalized ImFT for $\cont{0}$ maps.}
A generalized version of ImFT for continuous maps is presented by Halkin~\cite{halkin_0312017}. We can extend Proposition \ref{Imft_holtzman} to obtain estimates for such maps as follows. 
\begin{prop} 
\label{imft-c0}
Let $\Omega\subset\R[n]\times\R[m]$ be open and $(x_0,y_0)\in\Omega.$ Let $\Omega\ni(x,y)\mapsto P(x,y)$ be a continuous map. Assume: 
\begin{enumerate}[label=\textup{(\roman*)},leftmargin=*, widest=b, align=left]
\item $P(x_0,y_0) =0$;
\label{Apitem1}
\item $P$ is differentiable with respect to $y$ at $(x_0,y_0)$ and $\pdf{P}{y}(x_0,y_0)$ has a bounded inverse $\Gamma$ and $\norm{\Gamma}=k_1$. 
\label{Apitem2}
\item $S = \{(x,y)\in\ball{x_0}{\delta}\times\cl\ball{y_0}{\epsilon}\}\subset \Omega$.
\label{Apitem3}
\item there exists a real-valued function $[0,\delta]\times[0,\epsilon]\ni(u,v)\mapsto g_3(u,v)$ such that $g_3$ is nondecreasing in each argument with the other fixed and for all $(x,y)\in S$
\begin{equation*}
    \norm{P(x,y)-(P(x,y_0)+\pdf{P}{y}(x_0,y_0)\cdot(y-y_0))}\leq g_3(\norm{x-x_0},\norm{y-y_0});
\end{equation*}
\label{Apitem4}
\item there is a nondecreasing function $[0,\delta]\ni x\mapsto g_2(x)$ such that for all $(x,y_0)\in S$
\begin{equation*}
    \norm{P(x,y_0)}\leq g_2(\norm{x-x_0}); 
\end{equation*}
\label{Apitem5}
\item $k_1(g_2(\delta)+g_3(\delta,\epsilon))\leq \epsilon$.
\label{Apitem6}
\end{enumerate}
Then for all $x\in\ball{x_0}{\delta}$ there exists a $y_x$ (not necessarily unique) in $\cl\ball{y_0}{\epsilon}$ such that $P(x,y_x)=0$. 
\end{prop}
\begin{proof}
Existence of such a $g_1,~g_2$, $\epsilon>0$ and $\delta>0$ satisfying \ref{Apitem1}-\ref{Apitem6} is ensured from the continuity of $P$ and properties of the derivative $\pdf{P}{y}(x_0,y_0)$.
\\ \\
Fix an $x\in\ball{x_0}{\delta}$ and define
\begin{equation}
    \cl\ball{y_0}{\epsilon}\ni y \mapsto \P{x}{y} \coloneqq P(x,y)-P(x_0,y_0).
\end{equation}
then $P(x,y) = 0 \iff \P{x}{y} = P(x_0,y_0) = 0$. If $\P{x}{\cdot}$ vanishes on $\bd\ball{y_0}{\epsilon}$ then there is nothing to prove. Otherwise, assume $\P{x}{\cdot}$ is nonvanishing for all $y\in\bd\ball{y_0}{\epsilon}$. Define an approximation of $\P{x}{\cdot}$ as 
\begin{equation}
\label{appen1}
    \Papx{x}{y} = P(x,y_0) +\pdf{P}{y}(x_0,y_0)\cdot(y-y_0).
\end{equation}
Using \ref{Apitem2}, \ref{Apitem5} and \ref{appen1}, for all $y\in\bd\ball{y_0}{\epsilon}$ we have 
\begin{equation*}
    \norm{\Papx{x}{y}}\leq \frac{\epsilon}{k_1}-g_2(\delta).
\end{equation*}
Furthermore, from \ref{Apitem4}, for all $y\in\bd\ball{y_0}{\epsilon}$, we have
\begin{equation*}
    \norm{\P{x}{y}-\Papx{x}{y}}\leq g_3(\epsilon,\delta). 
\end{equation*}
From \ref{Apitem6}, for all $y\in\bd\ball{y_0}{\epsilon}$ we have 
\begin{equation*}
    \norm{\P{x}{y}-\Papx{x}{y}}\leq \norm{\Papx{x}{y}}.
\end{equation*}
Furthermore, $\Papx{x}{\cdot}$ and $\P{x}{\cdot}$ are nonvanishing on $\bd\ball{y_0}{\epsilon}$. Therefore, from Corollary \ref{cor_pbh} we have 
\begin{equation*}
  \DEG{\P{x}{\cdot}}{\ball{y_0}{\epsilon}} = \DEG{\Papx{x}{\cdot}}{\ball{y_0}{\epsilon}} \in\{-1,1\},   
\end{equation*}
and thus there exists a (not necessarily unique)  $y_x\in\cl\ball{y_0}{\epsilon}$ such that $\P{x}{y_x} = P(x,y_x) =0$.
\end{proof}
\end{appendices}

\end{document}